\documentclass[11pt,a4paper]{article}

\usepackage[colorlinks,pdfstartview=FitH,citecolor=red,linkcolor=red,urlcolor=red]{hyperref}

\usepackage{cite}
\usepackage{tikz}\usetikzlibrary{matrix,decorations.pathreplacing,positioning}
\usepackage{hyperref}
\usepackage{amsmath,amsthm,amssymb,bm,relsize}
\usepackage{calrsfs}
\usepackage{scalerel}

\usepackage{float}

\usepackage{pgfplots}
\pgfplotsset{compat=1.12}

\usepackage{setspace}
\usepackage{multirow}
\usepackage{array}
\newcolumntype{x}[1]{>{\centering\arraybackslash\hspace{0pt}}p{#1}}

\usepackage{empheq}
\usepackage{wasysym}

\usepackage{caption}
\usepackage{subcaption}

\usetikzlibrary{arrows.meta}

\usetikzlibrary{knots}
\newcommand{\wt}{\textnormal{wt}}
\usetikzlibrary{spath3}
\usetikzlibrary{hobby}

\usepackage{enumitem}
\setitemize{itemsep=-1pt}
\setenumerate{itemsep=-1pt}

\usepackage{titling}
\usepackage[margin=2.8cm]{geometry}
\usepackage{pgfplots}
\pgfplotsset{width=10cm,compat=1.9}

\definecolor{myg}{RGB}{220,220,220}

\usepackage[capitalise]{cleveref}

\theoremstyle{definition}
\newtheorem{theorem}{Theorem}[section]
\newtheorem{corollary}[theorem]{Corollary}
\newtheorem{proposition}[theorem]{Proposition}
\newtheorem{lemma}[theorem]{Lemma}

\newtheorem{definition}[theorem]{Definition}
\newtheorem{example}[theorem]{Example}
\newtheorem{notation}[theorem]{Notation}
\newtheorem{remark}[theorem]{Remark}
\newtheorem{terminology}[theorem]{Terminology}
\newtheorem{oproblem}[theorem]{Open Problem}

\newlength{\mynodespace}
\newcommand{\numberset}{\mathbb}

\newcommand{\F}{\numberset{F}}
\newcommand{\R}{\numberset{R}}
\newcommand{\Z}{\numberset{Z}}
\newcommand{\mA}{\mathcal{A}}
\newcommand{\mC}{\mathcal{C}}
\newcommand{\mD}{\mathcal{D}}
\newcommand{\mF}{\mathcal{F}}

\newcommand{\tub}{\textnormal{tub}}

\newcommand{\Ker}{\textnormal{Ker}}

\newcommand{\rk}{\textnormal{rk}}

\newcommand{\Mod}[1]{\ (\mathrm{mod}\ #1)}

\title{\textbf{Knot Theory and Error-Correcting Codes}}

\usepackage{authblk}

\author{Altan B. K\i l\i\c{c}, Anne Nijsten, Ruud Pellikaan, Alberto Ravagnani \\ 
Department of Mathematics and Computer Science, Eindhoven University of Technology, the Netherlands
\thanks{A. B. K\i l\i\c{c} is supported by the Dutch Research Council through grant VI.Vidi.203.045. and A. Ravagnani is supported by the Dutch Research Council through grants VI.Vidi.203.045, 
OCENW.KLEIN.539, 
and by the Royal Academy of Arts and Sciences of the Netherlands.}
\thanks{Emails: a.b.kilic@tue.nl, a.nijsten@live.nl, g.r.pellikaan@outlook.com, a.ravagnani@tue.nl.}}

\date{}

\begin{document}

\maketitle

\vspace{ 1 cm}
\begin{abstract}

This paper builds a novel bridge
between algebraic coding theory and mathematical knot theory, with applications in both directions. We give methods to construct error-correcting codes starting from the colorings of a knot, describing through a series of results how the properties of the knot translate into code parameters. We show that 
knots can be used 
to obtain error-correcting codes
with prescribed parameters 
and an efficient decoding algorithm. 
\end{abstract}

\medskip

\section*{Introduction}

The theory of error-correcting codes and their properties has been classically investigated in connection with several other areas of discrete mathematics, including finite geometry, enumerative combinatorics, algebraic combinatorics, algebraic and arithmetic geometry, matroid theory, ring theory, symbolic dynamics, and lattice theory to mention a few \cite{lind2021introduction, blake2014introduction,oxley2006matroid, stichtenoth2009algebraic, ball2015finite,ebeling2013lattices}. 

Studying codes in relation to other mathematical objects is 
an interesting and well-established 
research direction, which over the decades offered a new perspective on various classical problems. For example, deciding over which fields MDS codes exist is equivalent to deciding over which fields the uniform matroid is representable and is linked to the famous \textit{MDS Conjecture}~\cite{segre1955curve}. 

In this paper, we initiate the study of error-correcting codes in connection with mathematical knot theory, establishing a link between these two research domains. To our best knowledge, our paper is the first attempt to systematically and rigorously bridge coding theory with knot theory, except for the BSc and MSc theses of the second author of this paper~\cite{nijsten2019knots,nijsten2022knots}.

The way we associate codes to knots is via (Fox, Dehn or Alexander-Briggs) colorings of the \textit{knot diagram}. A knot diagram is a planar representation of a knot that can be divided into \textit{strands}, \textit{regions} and \textit{crossings}. These can be assigned \textit{colors}, which are elements of a commutative ring $R$ and where the coloring rules depend on some invertible element $t \in R$. The code is then constructed by using the \textit{coloring matrix} as a \textit{parity check matrix}; see Sections~\ref{sec:1} and~\ref{sec:2} for the definitions.

The paper then investigates how properties of knots translate into properties of the associated error-correcting code. To do so, we also establish some new properties of knot colorings. Most of our results focus on the length and the dimension of the associated code, but we are also able to 
prove some properties of the minimum distance (whose study appears to be a challenging task).

In our paper, we pay particular attention to two families of knots and their error-correcting codes.
These are \textit{torus knots} and their iterations, and \textit{pretzel knots}. We compute the parameters of the corresponding codes
in several instances. We also study the connected sum of knots and how the corresponding codes behave. We investigate the natural question of when the dual of a Fox knot code is a Fox knot code, and provide partial answers. 

\paragraph{Outline.} The remainder of this paper is organized as follows. In Section~\ref{sec:1}
we briefly review the preliminaries of knot and coding theory that are needed for this paper.
Section~\ref{sec:2} is about knot colorings and their algebra.
In Section~\ref{sec:3} we show how one can associate a code to a knot and investigate how the knot properties translate into code parameters. Section~\ref{sec:4} is devoted to torus knots, pretzel knots, and their associated codes. In Section \ref{sec:graph}, we study codes from graphs of Tait diagram of knots. Sections~\ref{sec:5} and~\ref{sec:6} conclude the paper and are about 
the connected sum of knots and the dual of Fox knot codes, respectively.
The paper also contains an appendix for the needed commutative algebra background.

\section{Knots and Codes}
\label{sec:preliminaries}
\label{sec:1}
In this section we give preliminary definitions and results on knot and coding theory that will be used throughout the paper. 
Since these two research areas are almost disjoint, we review the very basic concepts and include a selection of standard references. 
We assume that the reader is familiar with  
elementary concepts from algebra and topology; see~\cite{lang2012algebra} and~\cite{munkres} as standard references, among many others.

\subsection{Knot Theory}
\label{subsec:knot}
We start with the definition of a mathematical knot, following to various degrees~\cite{crowell2012introduction,livingston1993knot,kosniowski1980first,murasugi1996knot}. 

\begin{definition}
\label{def:knot}
A (\textbf{mathematical}) \textbf{knot} $K$ is a topological subspace of the Euclidean space~$\R^3$ that is homeomorphic to the unit circle $S^1 \subseteq \R^2$, endowed with the induced Euclidean topology. An \textbf{oriented knot} is the image of the unit circle under this map whose orientation is induced by the orientation of $S^1$ (clockwise or counterclockwise).
Knots $K_1, K_2 \subseteq \R^3$ are 
\textbf{equivalent} if there exists an orientation-preserving homeomorphism
$f: \R^3 \to \R^3$ such that
$f(K_1)=K_2$. A knot $K$ is called \textbf{trivial}~(or \textbf{unknotted}) if it is equivalent to the knot $$\{(x_1,x_2,0) \mid x_1, x_2 \in \R, \, x_1^2+x_2^2=1\} \subseteq \R^3.$$
\end{definition}

Making the notions of orientation
and orientation-preserving map
rigorous is 
a non-trivial task that is best accomplished by homology theory in algebraic topology; see e.g.~\cite[Chapter~22]{greenberg2018algebraic}.
Intuitively (and not rigorously), a homeomorphism $\R^3 \to \R^3$ is orientation-preserving if it sends a right-hand frame into a right-hand frame.
It can be shown (see~\cite[page~212]{kosniowski1980first}) that knots $K_1, K_2 \subseteq \R^3$ are equivalent if and only if there exists a homeomorphism~$f: \R^3 \to \R^3$ and a real number $\xi >0$ such that 
$f(K_1)=K_2$ and $f(x)=x$ for all~$x \in \R^3$ with $\left\| x\right\| \ge \xi$. The latter can be taken as an elementary, but fully rigorous, definition of a knot equivalence.

A trivial knot is also called an \textbf{unknot}. An unknot is depicted in Figure \ref{fig:unknot} and a \textbf{figure-eight} knot
is depicted in Figure~\ref{fig:fig8}. The latter is a non-trivial knot as we will explain later via colorings, see Figure \ref{fig:foxdehn}.

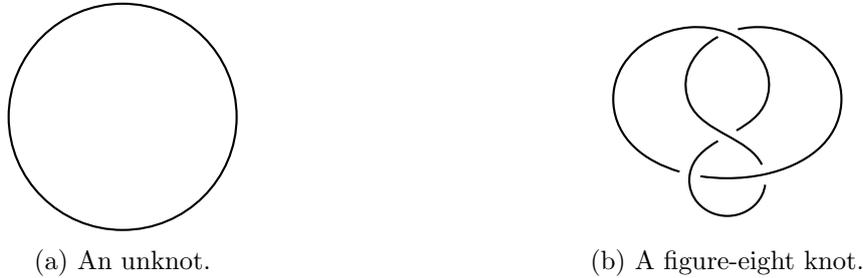
\begin{figure}[h!]
    \centering
\begin{subfigure}{.4\textwidth}
\centering
\begin{tikzpicture}
\filldraw[color=black, fill=white!5, thick](-1,0) circle (1.5);
\end{tikzpicture} 
    \caption{An unknot.}
    \label{fig:unknot}
\end{subfigure}
\hspace{0.1\textwidth}
\begin{subfigure}{.4\textwidth}
\centering
 \begin{tikzpicture}[use Hobby shortcut]
    \path[spath/save=figure8]
    ([closed]0,0) .. (1.5,1) .. (.5,2) ..
    (-.5,1) .. (.5,0) .. (0,-.5) .. (-.5,0) ..
    (.5,1) .. (-.5,2) .. (-1.5,1) .. (0,0);
    \tikzset{
    every spath component/.style={thick, draw},
    spath/knot={figure8}{8pt}{1,3,...,7}
    }
    \path (0,-.7);
    \end{tikzpicture}
    \caption{A figure-eight knot.}
    \label{fig:fig8}

\end{subfigure}

\caption{An example of a trivial and a non-trivial knot.}
\label{fig:exampleknots}

\end{figure}

A knot can sometimes be seen as an entangled polygon in a three-dimensional space. To make this formal, we give the following definition.

\begin{definition}
\label{def:polygonal}
A knot is called \textbf{polygonal} if it is a union of finite number of line segments. These line segments are the \textbf{edges} and their endpoints are the \textbf{vertices} of the knot.
\end{definition}

Note that the drawings of Figure \ref{fig:exampleknots} are smooth, but can be seen as polygonal knots with smoothened vertices. 
A knot that is equivalent to a polygonal knot is called \textbf{tame}. A knot that is not tame is called \textbf{wild}; see~\cite[Chapter~I]{crowell2012introduction}.

\begin{terminology} \label{term}
    In this paper, a knot will always mean an oriented, polygonal knot, unless otherwise stated. We will omit information about the orientation when it is not relevant.    
    See Remark~\ref{rem:explain} for the reason of restricting ourselves to this specific family of knots.
    Throughout this paper, $K$ always denotes a knot, unless otherwise stated.
\end{terminology} 

\begin{figure}[ht!]
    \centering
\begin{tikzpicture}[
use Hobby shortcut,
every trefoil component/.style={thick, draw},pics/arrow/.style={code={%
  \draw[line width=0pt,{Computer Modern Rightarrow[line
  width=1pt,width=3ex,length=2ex]}-] (-0.5ex,0) -- (0.5ex,0);
  }}]

\path[spath/save=trefoil] ([closed]90:2) foreach \k in {1,...,3} {
.. (-30+\k*240:.5) .. (90+\k*240:2) } (90:2);
\tikzset{spath/knot={trefoil}{8pt}{1,3,5}}
\node[text=black] at (0,2) {$\boldsymbol{>}$};
\end{tikzpicture}

    \caption{An oriented trefoil knot.}
    \label{fig:trefoil}
\end{figure}
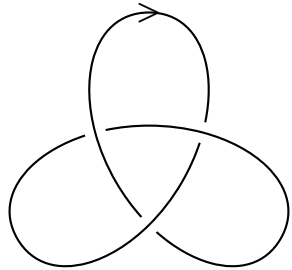

The knots that lie on the surface of an unknotted torus are of particular interest and will be used later in Section \ref{subsec:torus}. 

\begin{definition} \label{def:torusknot}
Consider the embedding of the torus $S^1\times S^1$  in $\R^3$ given by the implicit equation $$\left(\sqrt{x^2+y^2}-2\right)^2+z^2=1.$$
Let $(a,b)$ be a pair of nonzero integers that are relatively prime.
The $(a,b)$\textbf{-torus knot} $T(a,b)$ 
is the image of the map $S^1 \rightarrow \R^3$, lying on the torus, given by
$$
(\cos(t),\sin(t)) \longmapsto (\cos(at)(2+\cos(bt)),\ \sin(at)(2+\cos(bt)),\ \sin(bt));
$$ 
see e.g.~\cite[Chapter~7]{murasugi1996knot}.
The knot ``turns'' $a$ times meridionally and $b$ times longitudinally.
\end{definition}

\begin{example}
\label{ex:trefoil}
The {torus knot} $T(2,3)$ is more commonly known as the \textbf{trefoil knot}.  
It is depicted in Figure~\ref{fig:trefoil}. For any nonzero integer $a$, the torus knot
$T(a,\pm 1)$ is a trivial knot. The torus knots are completely classified; see \cite[Theorem 7.4.3]{murasugi1996knot}.

Figure \ref{fig:torus-poly}  depicts
the trefoil knot of Figure \ref{fig:trefoil} as an entangled polygon in a three-dimensional space, and as a knot that lie on the surface of a torus.

\begin{figure}[h!]
    \centering
\begin{subfigure}{.3\textwidth}
\centering
\begin{tikzpicture}[thick,scale=0.8]

\begin{knot}[
clip width=5,
flip crossing=1,
]
\strand[black,ultra thick] (2.1,4.2) to  (3,6) to (5,2) to (1.25,2);
\strand[black,ultra thick] (0.75,2) to  (-1,2) to (1,6) to (2.9,2.2);
\strand[black,ultra thick] (3.1,1.8) to  (4,0) to (0,0) to (1.9,3.8);
\end{knot}

\end{tikzpicture}
    \caption{A polygonal trefoil knot.}
    \label{fig:poly_trefoil}
\end{subfigure}
\hspace{0.1\textwidth}
\begin{subfigure}{.5\textwidth}
\centering
 \begin{tikzpicture}[thick,scale=0.8]
 \begin{axis}[axis equal image, hide axis, view = {152}{60}, scale= 2.8]

 \addplot3[domain=0:360, y domain=0:360, samples=20, surf, z buffer=sort]
 (
  {(2 + cos(x))*cos(y)},
  {(2 + cos(x))*sin(y)},
  {sin(x)}
 );

 \addplot3[domain=0:360, samples=50]
 (
 {(2 + cos(3*x))*cos(2*x)},
 {(2 + cos(3*x))*sin(2*x)},
 {sin(3*x)}
 );
 \end{axis}
\end{tikzpicture}    \caption{The torus knot $T(2,3)$.}
    \label{fig:tor_trefoil}

\end{subfigure}

\caption{The trefoil knot as an entangled polygon and as a torus knot.}
\label{fig:torus-poly}
\end{figure}
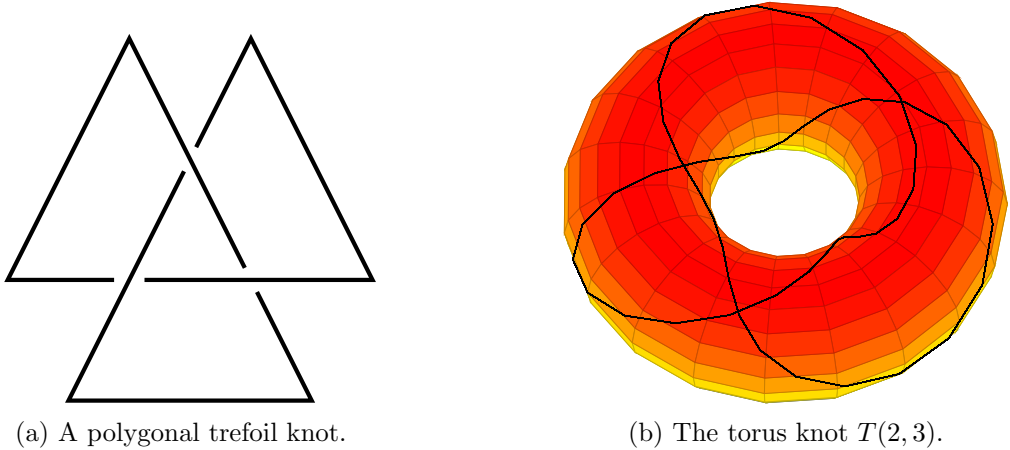
\end{example}

As in Figures \ref{fig:exampleknots} and \ref{fig:trefoil}, to visualize knots their two-dimensional projections are used. These are called knot \textit{diagrams} and  are defined as follows. We refer to~\cite{murasugi1996knot} for a complete treatment.

\begin{definition}
\label{def:knotdiag}
Let $p:\R^3 \rightarrow \R^3$ be defined by $p(x,y,z) = (x,y,0)$. 
The 
\textbf{projection}
of a polygonal knot $K$ is 
$p(K)$, together with the orientation inherited by
$K$, if $K$ was oriented.
The projection is called \textbf{regular} if it satisfies the following three conditions:
\begin{enumerate}
    \item $p(K)$ has at most a finite number of points of intersection, where $Q$ is a point of intersection of $p(K)$ if $|p^{-1}(Q)| > 1.$ 
    \item If $Q$ is point of intersection of $p(K)$, then $K \cap p^{-1}(Q)$ has exactly two points. Such a point is called a \textbf{double point} of $p(K)$.
    \item A vertex of $K$ is not mapped to a double point of $p(K)$.
\end{enumerate}

At a double point of a projection, to distinguish whether the
knot passes over or under itself, we draw the projection
so that it appears to have been cut; see for example Figure~\ref{fig:trefoil}. Such an altered projection is called a \textbf{diagram} of $K$. 
\end{definition}

From now on, we assume that the diagrams of knots we use in the paper are regular. This can be justified by the following theorem.

\begin{theorem} [see \cite{crowell2012introduction}]
\label{thm:reg}
Any polygonal knot $K$ is equivalent, under an arbitrarily small rotation of $\R^3$, to a polygonal knot $K'$ for which $p(K')$ is regular.
\end{theorem}

Thus, for a given polygonal knot there exists an equivalent knot with a regular projection. Combining with the definition of a tame knot, we have that every tame knot is equivalent to a polygonal knot with a regular diagram.

\begin{remark}
\label{rem:explain}
In knot theory, knots are studied up to equivalence. Most knot theory references focus solely on tame knots; see \cite{crowell2012introduction}. One of the reasons is that some very natural invariants are not necessarily defined for wild knots.
\end{remark}

Although we work with polygonal knots, their diagrams are depicted with smooth vertices, since one can think of a polygonal knot as a union of a large number of  edges. In the next definition, we introduce some terminology of knot diagrams.

\begin{definition} 
Each double point of a regular projection is the image of two different points of the knot, and a such a point is called a \textbf{crossing} of a diagram. To distinguish edges that cross each other in a diagram, the lower edge in the crossing is drawn with a break. The resulting separate edges are called \textbf{strands}. At each crossing, the strands that are separated by the break are called the \textbf{understrands} and the other strand is called the \textbf{overstrand}. The connected components of the complement of $p(K)$ in the $z=0$ plane are called the \textbf{regions}. 

\end{definition}

As an example, the diagram depicted in Figure \ref{fig:fig8} has 4 crossings and 4 strands, and the diagram of Figure \ref{fig:trefoil} has 3 crossings and 3 strands. It is not a coincidence that the number of crossings is equal to the number of strands. We now give a simple but fundamental lemma which will play an important role in the next section, where we explain knot colorings. The result can be found in~\cite{alexander1928topological}.

\begin{lemma}
\label{lem:knot_diag}
Let $D$ be knot diagram with $n$ crossings. Then it has $n$ strands and $n+2$ regions.
\end{lemma}

Elementary knot moves lead to changes in knot diagrams. However, it is possible to
restrict only to the following moves.

\begin{definition}
\label{def:reidemeister}
Consider the following three \textbf{Reidemeister moves}:
\begin{enumerate}[label= \Roman*.]
\item The \textbf{twist} move: This move twists or untwists a part of a strand in either direction, and is called a move of type I.
\item  The \textbf{poke} move: This move takes a strand and moves it completely over another (thus adding 2 crossings) or vice versa (thus removing 2 crossings), and is called a move of type II. 
\item  The \textbf{slide} move: This move slides a strand from one side of a crossing to the other side of the same crossing, and is called a move of type III.
\end{enumerate}
\end{definition}

The Reidemeister moves are depicted in Figure~\ref{fig:reidemeistermoves} and they are used to define equivalence of diagrams.

\begin{definition}
Two diagrams $D$ and $D'$ are called \textbf{equivalent} if $D$ can be transformed into~$D'$ by using a finite sequence of Reidemeister moves. We denote this by~$D \approx D'$.
\end{definition}

Reidemeister proved the following crucial result in \cite{reidemeister1927}. In this paper we use the statement of~\cite[Theorem 4.1.1]{murasugi1996knot}.

\begin{theorem}\label{thm:reidemeistermoves}
Let $D$ and $D'$ be the diagrams of two knots $K$ and $K'$, respectively. Then~$K \approx K'$ if and only if $D \approx D'$.
\end{theorem}

\begin{figure}[H]
    \centering
    \includegraphics[width=0.9\textwidth]{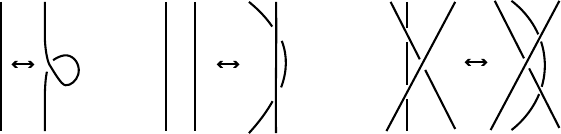}
    \caption{The Reidemeister moves of type I, II and III, respectively.}
    \label{fig:reidemeistermoves}
\end{figure}

Using the Reidemeister moves, one can show the equivalence of knots by applying Theorem~\ref{thm:reidemeistermoves}. For example, in Figure \ref{fig: amphichiral} we show that the figure-eight knot of Figure \ref{fig:fig8} is equivalent to its mirrored image. In the last step, no Reidemeister moves are used, but the position of the strands are changed slightly. The colors indicate how the strands are moved in the last step.
\begin{figure}[H]
    \centering
    \includegraphics[width=\textwidth]{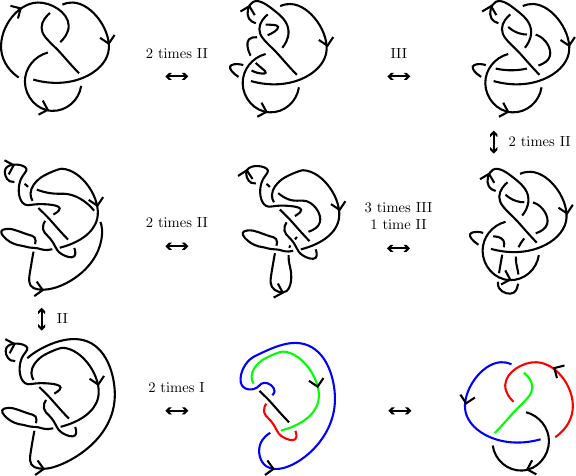}
    \caption{The figure-eight knot is equivalent to its mirror image.}
    \label{fig: amphichiral}
\end{figure}

In this paper, we will also use the concepts of a reduced and alternating 
knot diagrams.
These are defined as follows.

\begin{definition}
\label{def:reduced_alt}
A knot diagram is called \textbf{alternating} if the understands and overstrands are alternating in a fixed orientation. A knot diagram is called \textbf{reduced} if there are no crossings that can be removed via the twist move.
\end{definition}

\subsection{Coding Theory}
\label{subsec:coding}

We now turn to the coding theory fundamentals, that are also 
needed to understand the rest of the paper. Let $n \in \Z_{\ge 1}$, $q$ be a prime power, and $\F_q$ be the finite field with $q$ elements. 
General coding theory references are \cite{macwilliams1977theory,huffman2010fundamentals,pellikaan2018}.

\begin{definition}
A (\textbf{linear}, \textbf{error-correcting}) \textbf{code} of
\textbf{length} $n$ is
an $\F_q$-linear subspace~$\mC \subseteq \F_q^n$.
The \textbf{dimension} of $\mC$, denoted by $\dim(\mC)$, is its 
dimension as a vector space over~$\F_q$.
The quantity $\dim(\mC)/n$ is the \textbf{rate} of $\mC$, and denoted by $R(\mC)$.
The \textbf{dual} of $\mC$ 
is the code~$\mC^\perp = \{x \in \F_q^n \mid xy^\top =0 \mbox{ for all } y \in \mC\}$. Note that $\dim(\mC^\perp) = n - \dim(\mC)$.
A code $\mD \subseteq \mC$ is called a \textbf{subcode} of $\mC$.
\end{definition}

A code is most often represented by matrices.

\begin{definition}
\label{def:parity}
We say that a matrix $G \in \F_q^{k \times n}$ is a \textbf{generator matrix} of a code $\mC$ (and that $\mC$ is \textbf{generated} by $G$) if $\mC$ is the row-space of $G$. 
A \textbf{parity check matrix} $H$ of~$\mC \subseteq \F_q^n$ is a matrix such that
$$\mC=\{c \in \F_q^n \mid cH^T=0\}.$$
\end{definition}

Note that, in contrast with some coding theory references, we do not require $G$ and $H$ to have full rank in this paper. 

The performance of an 
error-correcting code is measured by its rate and its \textit{minimum Hamming distance}, defined below. Ideally, these parameters should both be as large as possible.

\begin{definition}
\label{def:weight_distance}
The \textbf{support} of a vector $x \in \F_q^n$ is
$\sigma(x) = \{i \in \{1,\ldots,n\} \mid \ x_i \neq 0\}$. The \textbf{Hamming weight} of a vector $x \in \F_q^n$ is the number of its nonzero entries, i.e.,~$\wt(x)=|\sigma(x)|$. 
The \textbf{minimum} (\textbf{Hamming}) \textbf{distance} of a code $\mC $
is 
$$d(\mC)=\min\left\{\wt(x) \mid x \in \mC, \, x \neq 0\right\},$$
where the code $\{0\} \subseteq \F_q^n$ has minimum distance $\infty$ by definition. The quantity $d(\mC)/n$ is the \textbf{relative minimum distance} of $\mC$, and denoted by $\delta(\mC)$.

The \textbf{weight enumerator} of $\mC $ is given by  
$W_{\mC}(t)=\sum_{w=0}^n a_w(\mC)t^w$, where $a_w(\mC)$ is the number of codewords of $\mC$ of weight $w$. 
Lastly, we let $\wt(\mC) = \{\wt(c) \mid c \in \mC \}$. 

\end{definition}

\begin{example}
\label{ex:repetition}
The $q$-ary \textbf{$n$-repetition code} is $\{(a,\ldots,a) \in \F_q^n \mid a \in \F_q\}$. It has dimension 1, minimum distance $n$, and rate $1/n$.

\end{example}

We write that $\mC$ is an $[n,k,d]_q$ code
if $\mC \subseteq \F_q^n$ has dimension $k$ and minimum distance~$d$. 

One of the best known results in coding theory establishes a trade-off between the dimension and the minimum distance of a code of a given length. In particular, they cannot be both arbitrarily large.

\begin{theorem}[\textbf{Singleton Bound}; see~\cite{singleton1964maximum}] \label{thm:sbound}
Let $\mC \neq \{0\}$ be an $[n,k,d]_q$ code. We have~$k \le n-d+1$.
\end{theorem}
Another very famous bound is the Gilbert-Varshamov bound.
\begin{theorem}[see~\cite{gilbert1952comparison,varshamov1957estimate}] \label{thm:gvbound}
Let $\mC \neq \{0\}$ be an $[n,k,d]_q$ code. We have 
$$q^{n-k} \le \sum_{i=0}^{d-1} \binom{n}{i}(q-1)^i.$$
\end{theorem}

Next, we give three definitions of code equivalence.
\begin{definition}
\label{d-mono-perm-pm-equiv}
Two $\F_q$-linear codes are called \textbf{permutation equivalent} if one is obtained from the other by permuting the coordinates.
They are called \textbf{monomial equivalent} if one is obtained from the other by permuting the coordinates and by multiplying the coordinates 
with nonzero elements of the field $\F_q$, see \cite{huffman2010fundamentals} and \cite[Definition 1.1.15]{pellikaan2018}.
They are called \textbf{($\pm1$)-permutation equivalent} if one is obtained from the other by permuting the coordinates and by multiplying the coordinates with $\pm 1$.
\end{definition}

Over $\F_2$, the three equivalences defined in \ref{d-mono-perm-pm-equiv} are the same.
Next, we define two classes of codes that will arise later in our paper. These two classes of codes are examples of well-known families from classical coding theory that can be obtained as knot codes, and thus have particular interest for us.

\begin{definition}
    \label{def:ldpc}
    A code that has a parity check matrix in which every row has Hamming weight $r$ and every column has Hamming weight $c$, is called a \textbf{$(r,c)$-doubly-regular low-density parity check (LDPC) code}. If the rows or  the columns of the matrix have a fixed Hamming weight $w$, then the LDPC code is called \textbf{right} or \textbf{left $w$-regular}, respectively.
\end{definition}

LDPC codes, first introduced in \cite{gallager1962low}, have efficient decoding algorithms, see for example~\cite{mackay1997near,mackay1999good, luby2001improved} among many others.

\begin{definition}
\label{def:lcd}
The \textbf{hull} of a code $\mC$ denoted by $H(\mC)$ is the intersection of the code with its dual: 
 $H(\mC)=\mC \cap \mC^\perp$. A code $\mC$ is called \textbf{linear complementary dual (LCD)} if $H(\mC)=\{ 0 \}$. See \cite{massey1992}.
\end{definition}

 LCD codes have been widely applied in data storage, communications systems, consumer electronics, and cryptography \cite{carlet2016}. 

\begin{definition}
\label{def:asympgood}
A sequence of linear codes $(\mC_j)^{\infty}_{j=1}$ where each $\mC_j$ has parameters $[n_j,k_j,d_j]$ is called \textbf{asymptotically good} if the following hold:
\begin{enumerate}
    \item $\lim_{j \to \infty}n_j = \infty$,
    \item $\liminf_{j\to\infty} R(\mC_j) > 0$,
    \item $\liminf_{j\to\infty} \delta(\mC_j) > 0$.
\end{enumerate}

\end{definition}

\section{Knot Colorings}
\label{subsec:color}
\label{sec:2}

In this section we explain three types of knot colorings. 
Fox coloring and Dehn coloring are colorings of the strands and crossings, respectively, see Figure \ref{fig:Fox-Dehn}. For the third, the Alexander-Briggs coloring, we first define the \textit{Tait diagram} of an oriented knot, see Figure \ref{fig:Tait}. 
We refer to Appendix~\ref{commut-alg} for the necessary background in commutative algebra needed for this section.

\begin{figure}[h!]
    \centering
\begin{subfigure}{.4\textwidth}
\centering

\begin{tikzpicture}[pics/arrow/.style={code={%
  \draw[line width=0pt,{Computer Modern Rightarrow[line
  width=1pt,width=3ex,length=2ex]}-] (-0.5ex,0) -- (0.5ex,0);
  }}]
\begin{knot}[
clip width=5,
]
\strand[black,ultra thick] (-2,-2) to pic[pos=0.2,sloped]{arrow} (2,2);
\strand[black,ultra thick] (-2,1) -- (2,-1);;
\end{knot}

\tikzset{nnode/.style = {shape=circle,fill=myg,draw,inner sep=0pt,minimum
size=1.9em}}
\node[text=red] at (-2.3,1) {$a$};
\node[text=red] at (2.3,-1) {$c$};
\node[text=red] at (-2.3,-2) {$b$};
\end{tikzpicture}
    \caption{$c=ta +(1-t)b$}
    \label{fig:Fox_coloring}
\end{subfigure}
\hspace{0.1\textwidth}
\begin{subfigure}{.4\textwidth}
\centering
\begin{tikzpicture}[pics/arrow/.style={code={%
  \draw[line width=0pt,{Computer Modern Rightarrow[line
  width=1pt,width=3ex,length=2ex]}-] (-0.5ex,0) -- (0.5ex,0);
  }}]
\tikzset{nnode/.style = {shape=circle,fill=myg,draw,inner sep=0pt,minimum
size=1.9em}}

\begin{knot}[
clip width=4,
]
\strand[black,ultra thick] (-2,-2) to pic[pos=0.2,sloped]{arrow} (2,2);
\strand[black,ultra thick] (-2,1) -- (2,-1);;
\end{knot}

\node[text=red] at (0,-1) {$U_i$};
\node[text=red] at (0,1) {$U_l$};
\node[text=red, ultra thick] at (-1.2,-0.2) {$U_j$};
\node[text=red] at (1,0) {$U_k$};

\end{tikzpicture}
    \caption{$U_i-tU_j=U_k-tU_l$}
    \label{fig:Dehn_coloring}
\end{subfigure}

\caption{Fox coloring \ref{fig:Fox_coloring} and Dehn coloring \ref{fig:Dehn_coloring} of knot diagrams. }
\label{fig:Fox-Dehn}
\end{figure}
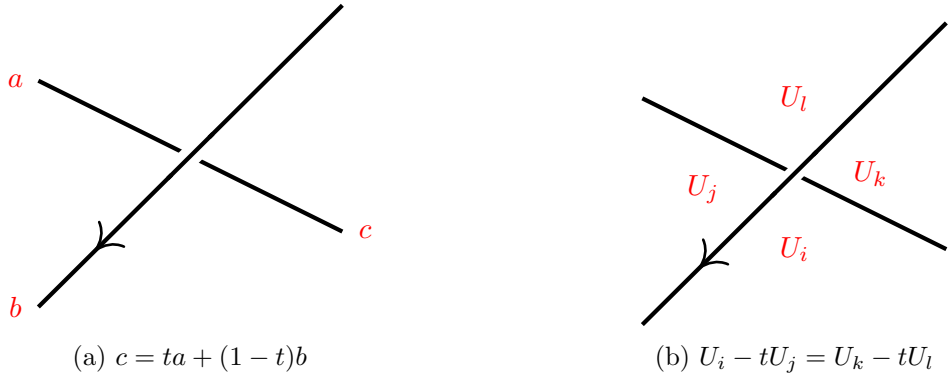

\subsection{Fox Coloring}

The Fox coloring is the coloring of the strands of the knot diagrams. In this section we introduce the concept of \textit{Fox $(R,t)$-coloring}, where $R$ is a Noetherian commutative ring with an identity and $t$ is an invertible element of this ring. 
We start with the definition of \textit{Fox tricolorability}, that is where $R=\Z/(3)$ and $t=-1$.

\begin{definition}
\label{def:3color}
A \textbf{Fox tricoloring} of a  knot diagram is a coloring of the strands with three colors such that at each crossing, the colors of the strands that meet at that crossing are either all the same or all different. If we take as colors $0$, $1$ and $2$, then this rules amounts to the linear equation $a+b+c \equiv 0 \Mod{3}$, where $a,b$ and $c$ are the colors of the three strands that come together at a crossing. Moreover, a Fox tricoloring is called \textbf{trivial} if all strands have the same color. A knot diagram is called \textbf{Fox tricolorable} if it has a non-trivial tricoloring.

\end{definition}
Tricolorability is another invariant of a knot\cite{przytycki}. This already allows us to distinguish the unknot and trefoil knot, as the latter is tricolorable and the former is not, see Figure \ref{fig:3color_trefoil}.

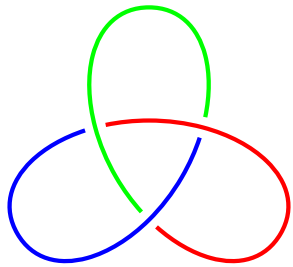
\begin{figure}[ht!]
    \centering
\begin{tikzpicture}[
use Hobby shortcut,
every trefoil component/.style={ultra thick, draw},
trefoil component 1/.style={red},
trefoil component 2/.style={blue},
trefoil component 3/.style={green},
]
\path[spath/save=trefoil] ([closed]90:2) foreach \k in {1,...,3} {
.. (-30+\k*240:.5) .. (90+\k*240:2) } (90:2);
\tikzset{spath/knot={trefoil}{8pt}{1,3,5}}
\end{tikzpicture}

    \caption{Trefoil knot is tricolorable.}
    \label{fig:3color_trefoil}
\end{figure}

Generalizing Definition \ref{def:3color} to colors $0,1, \ldots, n-1$ gives the equation $a+c \equiv 2b \Mod{n}$, where $a,b$ and $c$ are again the colors of the three strands that come together at a crossing with $b$ being the overstrand, and $n \in \Z_{>0}$. This can also be generalized further.

\begin{definition}
\label{def:FoxRt}
A \textbf{Fox $(R,t)$-coloring} of a knot diagram is a coloring of its strands with colors that are elements of $R$ and for each crossing it holds that
\begin{equation}
\label{eq:Fox}
c=ta +(1-t)b,
\end{equation}
where $t$ is a fixed invertible element in the ring $R$, the strand with color $b$ is the overstrand and the strands colored with $a$ and $c$ are understands such that the rotation from $b$ to $c$ around the crossing is counter clockwise; see Figure \ref{fig:Fox_coloring}.
A coloring is called trivial if all the colors are the same. The knot diagram is called \textbf{Fox $(R,t)$-colorable} if there is a non-trivial Fox $(R,t)$-coloring; see \cite{fox1970}.
\end{definition}

In particular, every Fox $(R,t)$-coloring with $R=\Z/(2)$ is trivial, since $t=1$ is the only invertible element of $R=\Z/(2)$. We also have the following result.

\begin{proposition}\label{p-(R,1)-col}
All Fox $(R,1)$-colorings of a knot diagram are trivial.
\end{proposition}

\begin{proof}
Up to a permutation, we may assume that the strands~$\{x_1,\ldots,x_n\}$ are numbered in such a way that $x_{j+1}$ comes after $x_j$ for a given choice of the orientation of the knot diagram.
Then the equations \eqref{eq:Fox} with $t=1$ become $x_{j+1}=x_j$ for all $j=1, \ldots n-1$. Hence every Fox $(R,1)$-coloring is trivial.
\end{proof}

Lemma \ref{lem:knot_diag} allows us to place the coefficients in the system of equations that has to hold for a Fox coloring of a diagram with $n$ strands into an $n \times n$ matrix. Before defining this matrix, we remark the following.

\begin{remark}\label{r-alex}
The definition of the \textit{Alexander matrix} of a knot diagram is usually given by means of the free calculus of a presentation of the fundamental group of the complement of the knot; see \cite[Chapter III]{crowell2012introduction}. From this approach one readily gets that the distinct presentations of the fundamental group of the knot give equivalent (see Definition \ref{d-equiv-matrix}) Alexander matrices. Hence the equivalence class of the Alexander matrix of a knot diagram is in fact an invariant of the knot.
\end{remark}

We give the following \textit{ad hoc} definition of the Alexander matrix of a knot diagram and show that it is an invariant under Reidemeister moves.

\begin{definition}
\label{def:coloringMat}
The \textbf{Alexander matrix} (or the \textbf{Fox coloring matrix}) of an oriented knot diagram with $n$ crossings $\{c_1,\ldots,c_n\}$ and strands~$\{x_1,\ldots,x_n\}$ is defined as the 
matrix~$M(t)$ with entries in $\Z[t,t^{-1}]$ with
$$
M_{ij}(t) = 
\begin{cases}
1 - t & \text{if} \ \ x_j \textnormal{ is an overstrand at } c_i,\\
-1 & \text{if} \ \ x_j \textnormal{ is an understrand at } c_i \textnormal{ at the left side of the overstrand},\\
t & \text{if} \ \ x_j \textnormal{ is an understrand at } c_i \textnormal{ at the right side of the overstrand},\\
0 & \text{otherwise}
\end{cases}
$$ for $1 \le i,j \le n$ with $i,j \in \Z_{>0}.$
\end{definition}

The matrix of Definition \ref{def:coloringMat} is called the Alexander matrix by Crowell and Fox \cite{crowell2012introduction} and it is different than the matrix Alexander defined in his paper~\cite{alexander1928topological}. 
Note that we write $M_{ij}(t)$ instead of 
$M(t)_{ij}$ and, for ease 
of notation, we omit the knot diagram in the symbol $M_{ij}(t)$.

\begin{definition}
\label{def:Foxmodule}
The \textbf{module of Fox $(R,t)$-colorings} of a knot diagram is the kernel of the matrix $M(t)$, that is, the $R$-module of 
$x \in R^n$ such that $M(t) x^T=0$.
\end{definition}

The sum of the entries in any row of $M(t)$ is zero. So, the columns of $M(t)$ are dependent, and thus the determinant of $M(t)$ is zero. 
Hence $E_0(M(t))=0$ (see Definition~\ref{d-elem-id}) and there is a non-trivial solution of the system of equations \eqref{eq:Fox}. Trivial colorings correspond to the solutions of this system of equations where all the (nonzero) elements are the same. 
By taking any $(n-1)$-minor of the Alexander matrix, we get another knot invariant; see~\cite{alexander1928topological}.

\begin{proposition}\label{p-pincipalM}
Let $M_{ij}^*(t)$ be the matrix obtained by deleting the $i$-th row and the $j$-th column of $M(t)$. 
Then the determinants $\det (M_{ij}^*(t))$ are equal to each other up to a factor~$\pm t^s$, where $s$ is an integer. 
In particular, $E_1(M(t))$ (see Definition \ref{d-elem-id}) is a principal ideal generated by $\det (M_{ij}^*(t))$ for any choice of the indices $1 \le i, j \le n$.
\end{proposition}

\begin{proof}
See \cite[Chapter VI (1.3)]{crowell2012introduction} and \cite[Chapter VIII (3.7)]{crowell2012introduction}.
\end{proof}

\begin{proposition}\label{p-det=pm1}
Let $M(t)$ be the Alexander matrix of a knot diagram with $n$ crossings. We have $\det (M_{ij}^*(1))=\pm1$ for all $1 \le i \le j \le n.$
\end{proposition} 
\begin{proof} 
The equations \eqref{eq:Fox} with $t=1$ become $x_{j+1}=x_j$ for all $j=1, \ldots,n-1$ as we have seen in the proof of Proposition \ref{p-(R,1)-col}. The matrix that is obtained by deleting the last column and last row is upper triangular with ones on the diagonal, so it has determinant one. The result follows from Proposition \ref{p-pincipalM}.
\end{proof}

Another important result is the following.

\begin{proposition}\label{p-ReidemesterM}
If $D_1 \approx D_2$, 
then the corresponding Alexander matrices~$M_1(t)$ and~$M_2(t)$ are equivalent; see Definition \ref{d-equiv-matrix}. 
\end{proposition}

\begin{proof}
See \cite[Chapter 2, Theorem 3]{livingston1993knot} in the case $t=-1$ and $R=\F_p$ for a prime $p$. The general case is proved similarly.
\end{proof}

$E_1(M(t))$ is a principal ideal in $\Z[t,t^{-1}]$ by Proposition \ref{p-pincipalM}, that is generated by a nonzero element by Proposition \ref{p-det=pm1}.
Hence there exists an integer $s$ such that multiplication of a generator of $E_1(M(t))$ by the invertible element $\pm t^s$ gives a polynomial with a constant term that is positive. 

\begin{definition}
\label{def:alex}
Let $K$ be a knot. The \textbf{Alexander polynomial} of $K$, denoted by $\Delta_K(t)$, 
is the generator of $E_1(M(t))$ which is the polynomial with a constant term that is positive. Moreover, the value $|\Delta_K(-1)|$ is called the \textbf{determinant} of $K$. 
\end{definition}

The Alexander polynomial is a knot invariant by Proposition \ref{p-ReidemesterM}. However it is important to note that although the elementary ideals $E_0(M(t))=0$ and $E_1(M(t))$ of a knot are principal ideals, the other elementary ideals $E_k(M(t))$ are not necessarily principal for~$k>1$, as the following example shows.

\begin{example}
\label{ex:not-principal-ideal}
Figures 50 and 51 of \cite{crowell2012introduction} have both $\Delta_K(t)=2t^2-5t+2$ as their Alexander polynomial, but they have distinct elementary ideals.
The Stevedore's knot depicted in Figure~50 has $E_k=(1)$ for all $k\geq 2$, but the knot of Figure~51 has $E_2=(2-t,1-2t)$, which is not principal.
\end{example}

We work out an example to show how the Alexander polynomial of a knot is computed. Note that it does not depend on the chosen submatrix or the chosen knot diagram.

\begin{example}
\label{ex:alexander}
The Alexander matrix of the diagram of the trefoil knot $K$ depicted in Figure \ref{fig:trefoil} is given by 
\begin{equation*}
M(t) = \begin{pmatrix}
1-t & t & -1 \\
-1 & 1-t & t \\
t & -1 & 1-t
\end{pmatrix}.
\end{equation*}
We have $\det(M^*_{11}(t)) = t^2 -t + 1$ and $\det(M^*_{12}(t)) = -t^2 +t - 1$. Following Definition~\ref{def:alex}, we observe that the polynomial $\det(M^*_{11}(t))$ has a positive constant term, and thus~$\Delta_K(t) = t^2-t+1$.
\end{example}

We now turn our attention to the invariant factors of  the Alexander matrix and the invariant factors of the module of Fox $(R,t)$-colorings, see Definition \ref{def:Foxmodule}.

\begin{proposition}\label{p-elem-ideals-Alex-pol}
Let $R$ be a principal ideal domain with invertible element $t$. 
Let~$(d_1)\subseteq (d_2)\subseteq \cdots \subseteq (d_l)$ be the  invariant factors of the matrix $M(t)$ and~$E_k(M(t))$ be generated by $\Delta_k$, see Corollary \ref{c-struct-pid}. Then $\Delta_0=0$, $d_1=0$, and $\Delta_K(t)=\Delta_1=\prod_{j=2}^ld_j$. 
\end{proposition}

\begin{proof}
The determinant of $M(t)$ is zero. So $E_0(M(t))=(0)$ and $\Delta_0=0$.
Now~$\Delta_K(t)=\Delta_1=\prod_{j=2}^ld_j$ by Corollary \ref{c-struct-pid} and $\Delta_K(t)\not=0$ by Proposition \ref{p-det=pm1}. 
So $\Delta_0=d_1\Delta_1$ by Corollary \ref{c-struct-pid}. This implies that $d_1=0$, since $R$ is an integral domain.
\end{proof}

Slightly abusing the notation, denote the localization of $\Z$ at a nonzero $t\in \Z$ by $Z_t$, and the localization of $\F_p[T]$ at a nonzero $t$ of $\F_p[T]$ by $\F_p[T]_t$. So $Z_t=\{ n/t^e \mid n,e \in\Z \}$ and~$\F_p[T]_t=\{ f/t^e \mid f\in \F_p[T],\ e \in\Z \}$. The next two propositions will be used later in Subsection \ref{subsec:dim} to bound the dimension of so-called \textit{Fox knot codes}.

\begin{proposition}\label{p-number-color-Z}
Let $D$ be a  knot diagram of a knot $K$.
Let $d, t\in \Z$ such that $0<t<d$ and $ \gcd(d,t)=1$. Let $R=\Z_t$ and and $\overline{R}=R/(d)$.
Let  $(d_1)\subseteq (d_2)\subseteq \cdots \subseteq (d_l) $ be the invariant factors of the matrix $M(t)$ of Fox $(R,t)$-colorings.
Let $a_i=\gcd (d,d_i)$ and~$\overline{x}=x+(d) \in R/(d)$ for $x\in R$.
Then $\overline{t}$ is an invertible element in $\overline{R}$ and $d\prod_{i=2}^n a_i$ is the number of Fox $(\overline{R},\overline{t})$-colorings of $D$.
\end{proposition}
\begin{proof}

The element $\overline{t}$ is invertible in $R/(d)$, since $ \gcd(d,t)=1$. 
Hence~$\Z/(d)\cong \Z_t/(d)=\overline{R}$.
The module of $(\overline{R},\overline{t})$-colorings of $D$ is equal to the $Ker(M(\overline{t}) )$ which is isomorphic to~$ \overline{R}/(\overline{a_1}) \oplus \overline{R}/(\overline{a_2}) \oplus  \cdots \oplus  \overline{R}/(\overline{a_l})$ by Proposition~\ref{p-struct-ker-im}. We have $d_1=0$ by Proposition~\ref{p-elem-ideals-Alex-pol}, and $|\overline{R}|=d$. Furthermore $\overline{R}/(\overline{a_i})\cong \Z/(a_i)$, and thus $|\overline{R}/(\overline{a_i})|=a_i$ for all $2 \leq i\leq n$.
Hence~$d\prod_{i=2}^n a_i$ is the number of Fox $(\overline{R},\overline{t})$-colorings of $D$.
\end{proof}

\begin{proposition}\label{p-number-color-prime}
Let $D$ be a  knot diagram of a knot $K$.
Let $p\in \Z$  be a prime number. Let $d,t \in \F_p[T]$ such that $\gcd(d,t)=1$. Let $R=\F_p[T]_t$ and let $\overline{R}=R/(d)$.
Let~$(d_1)\subseteq (d_2)\subseteq \cdots \subseteq (d_l) $ be the invariant factors of the matrix $M(t)$ of Fox $(R,t)$-colorings.
Let~$a_i=\gcd (d,d_i)$. Let $\delta = \deg (d)$ and $\alpha_i = \deg (a_i)$.
Then $\overline{t}$ is an invertible element of $\overline{R}$ and~$p^{\delta+\sum_{i=2}^n \alpha_i}$ is the number of Fox $(\overline{R},\overline{t})$-colorings of $D$.
\end{proposition}
\begin{proof}
The proof is verbatim the same as for Proposition \ref{p-number-color-Z}, except for the final part. The element $\overline{t}$ is invertible in $R/(d)$, since $ \gcd(d,t)=1$. 
Hence $\overline{R}\cong \F_p[T]/(d)$ which has $p^{\deg(d)}$ elements, and 
$\overline{R}/(\overline{a_i})\cong \F_p[T]/(a_i)$ which has $p^{\deg(d_i)}$ elements.
\end{proof}

The Alexander polynomial~$\Delta_K(t)$ plays a direct role in determining whether a knot diagram is Fox $(R,t)$-colorable or not.

\begin{proposition}
\label{prop:determinant}
Let $R=\Z$ or $R=\F_q[T]$. 
Let $d,t \in R$ such that $d$ is not invertible in~$R$ and $\gcd(d,t)=1$.
Let $\overline{R}= R/(d)$ and $K$ be a knot.
Then the following statements are equivalent:\\
(1) A knot diagram of $K$ is Fox $(\overline{R},t)$-colorable;\\
(2) $\gcd(d,\Delta_K(t))\not=1$ in $R$;\\
(3) $\Delta_K(t)=0$ in $\overline{R}$ or  $\Delta_K(t)$ is a zero-divisor of $\overline{R}$.
\end{proposition}
\begin{proof}
The proof  Fox $(\F_p,-1)$-colorability for $p$ a prime is given in \cite[Chapter 3, Theorem 4]{livingston1993knot} and \cite[Proposition 2.1]{kauffman2018}. The knot $K$ is Fox $(\overline{R},t)$-colorable if and only if 
$\gcd (d,d_i)=\overline{d_i}\not=1$ for some $i$, $2 \leq i\leq n$, where
$(d_1)\subseteq (d_2)\subseteq \cdots \subseteq (d_l)$ are the invariant factors in $R$ 
by Propositions~\ref{p-number-color-Z} and~\ref{p-number-color-prime}.
But $\Delta_K(t)=\prod_{i=2}^l d_i$ by Proposition \ref{p-elem-ideals-Alex-pol}.
So $K$ is Fox $(\overline{R},t)$-colorable if and only if $\gcd(d,\Delta_K(t))\not=1$ in $R$ if and only if
$\Delta_K(t)=0$ in $\overline{R}$ or  $\Delta_K(t)$ is a zero-divisor of $\overline{R}$.
\end{proof}

Next, we will show that the trefoil knot is Fox $(R,t)$-colorable for several choices of the ring $R$ and the invertible element $t$.
\begin{example}
\label{ex:Z4F4}
We have seen that the trefoil knot is tricolorable, which is in agreement with Proposition \ref{prop:determinant}, since $\Delta_K(t) = t^2 - t + 1$ and $\Delta_K(-1) =3$. It is also $(\Z/(d),-1)$-colorable for all positive integers $d$ that are a multiple of $3$ with the colors $0, d/3, 2d/3$ assigned to the three strands.

Consider the Fox colorings for the pairs $(\Z/(4), -1)$, $(\mathbb{F}_4, \alpha)$ and $(\mathbb{F}_7, 3)$ of the trefoil knot~$K$, where $\mathbb{F}_4 = \{0, 1, \alpha, \alpha^2 \}$ and $\alpha$ is a root of the irreducible polynomial $x^2+x+1$ over~$\F_2[x]$. We find that $\Delta_K(-1) = 3 \neq 0$ over $\Z/(4)$ and $\Delta_K(\alpha) = \alpha^2 + 1 - \alpha$, which is 0 over~$\mathbb{F}_4$ and $\Delta_K(3) = 7 = 0$ over $\mathbb{F}_7$. Therefore, the trefoil knot has only trivial Fox colorings when~$(R,t) = (\Z/(4), -1)$, but it has a non-trivial Fox coloring when $(R,t) \in \{(\mathbb{F}_4, \alpha),(\mathbb{F}_7, 3)\}$.
\end{example}

\subsection{Dehn Coloring}

In this subsection, we study another way to color knot diagrams, called Dehn colorings. The \textit{Dehn coloring} is the coloring of the regions of a knot diagram. Similar to Definition~\ref{def:FoxRt}, we give the definition of a Dehn coloring as follows. 

\begin{definition}
\label{def:DehnRt}

A coloring of the regions of a knot diagram with $n$ crossings is called a \textbf{Dehn $(R,t)$-coloring} 
if the regions are colored via colors that are elements of $R$ 
and at each crossing~$c_m$ for $1 \le m \le n$ with an overstrand $x$, it holds that
\begin{equation*}
U_i - tU_j = U_k - tU_l,
\end{equation*}
where $t$ is a fixed invertible element in the ring $R$, the regions $U_i,U_j,U_k$ and $U_l$ are regions that have $c_m$ on their border in a way that $U_i$ and $U_k$ are on the left side of~$x$ and~$U_j$ and~$U_l$ are on the right side of $x$ with respect to the orientation of the diagram; see Figure~\ref{fig:Dehn_coloring}. Following the convention, the color $0$ is assigned to the unbounded outside region. 
\end{definition}

Analogous to Definition \ref{def:coloringMat}, we define the following matrix for Dehn colorings.

\begin{definition}
\label{def:DehnMat}
Let $D$ be a knot diagram with $n$ crossings. At each crossing $c_m$ for~$1 \le m \le n$ with an overstrand $x$ such that the regions $U_i,U_j,U_k$ and $U_l$ are regions that have $c_m$ on their border in a way that $U_i$ and $U_k$ are on the left side of $x$ and $U_j$ and $U_l$ are on the right side of $x$ with respect to the orientation of the diagram, the \textbf{Dehn coloring matrix}~$N(t)$ of $D$ is defined as
$$
N_{ms}(t) = 
\begin{cases}
1 & \text{if} \ \ s=i,\\
-t & \text{if} \ \ s=j,\\
-1 & \text{if} \ \ s=k,\\
t & \text{if} \ \  s=l,\\
0 & \text{otherwise},
\end{cases}
$$ for $1 \le m \le n$ and $1 \le s \le n+2$ with $m,s \in \Z_{>0}.$
\end{definition}

The matrix $N(t)$ is the one defined by Alexander \cite{alexander1928topological} as remarked after Definition \ref{def:coloringMat}.  

\begin{remark}\label{r-dehn}
Dehn \cite{dehn1910} gave a less known presentation of the fundamental group of the complement of a knot with generators $U_i$
and relations $U_1=1$ and $U_iU_j^{-1}=U_kU_l^{-1}$ for all crossings as in Figure \ref{fig:Dehn_coloring}.
The free calculus of this presentation gives the matrix $N(t)$ with the first column deleted, see \cite{kauffman1983}.
 
\end{remark}

Analogous to Proposition \ref{p-ReidemesterM}, we have the following result that is proven in \cite{alexander1928topological}.

\begin{proposition}\label{p-ReidemesterN}
If $D_1 \approx D_2$, 
then the corresponding Dehn coloring matrices~$N_1(t)$ and~$N_2(t)$ are equivalent.
\end{proposition}

Similar to the module of Fox $(R,t)$-colorings of Definition \ref{def:Foxmodule}, one can define the module of Dehn $(R,t)$-colorings.

\begin{definition}

The \textbf{module of Dehn $(R,t)$-colorings} of the knot diagram is given by the kernel of the matrix $N(t)$, 
that is the $R$-module of all $ x\in R^{n+2}$  such that $N(t)x^T=0.$
\end{definition}

Fox and Dehn colorings can be obtained from each other. 
The following proposition is a generalization of the relation between Fox and Dehn colorings. We slightly abuse notation: the color of a region $U$ will also be denoted by $U$. Similarly, the color of a strand $x$ is also denoted by $x$.

\begin{proposition}
\label{prop:fox2dehn}
Let $D$ be an oriented knot diagram with $m$ regions and $n$ strands. 
Consider the map $\varphi : R^m \rightarrow R^n$ such that  $\varphi (U)=x$ 
gives the colors of the strands $x$ for a given coloring $U$ of the regions
such that $x_r =U_i-tU_j$ is the color of the stand $x_r$ where~$U_i$ and $U_j$ are the colors of the regions next to the strand~$x_r$, with $U_i$ on the left side of $x_r$  and $U_j$ on the right side of $x_r$. 
Then this map is a well-defined morphism of $R$-modules when restricted to the module of Dehn $(R,t)$-colorings, and Dehn colorings are mapped to Fox  $(R,t)$-colorings. 
Furthermore $\varphi$ is surjective onto the module of Fox  $(R,t)$-colorings, its kernel is isomorphic to $R$, 
and it is an isomorphism when $\varphi$ is restricted to the submodule of Dehn $(R,t)$-colorings where a fixed region gets the value $0$.
\end{proposition}

\begin{proof}
It is a straightforward generalization of the proofs given in \cite{carter2014three,traldi2017} for $t=-1$.
\end{proof}

In Figure \ref{fig:foxdehn} an example of a Fox $(\F_5,-1)$-coloring and Dehn $(\F_5,-1)$-coloring that are constructed via these steps can be found with the value $0$ for the outside region. 

\begin{figure}[H]
    \centering
    \includegraphics[width=0.5\textwidth]{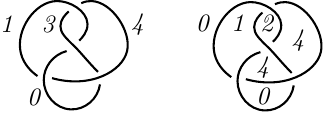}
    \caption{Fox $(\F_5,-1)$-coloring and Dehn $(\F_5,-1)$-coloring of the figure-eight knot.}
    \label{fig:foxdehn}
\end{figure}

\begin{remark}\label{r-index}Consider the diagram of an oriented knot. 
The \textbf{index} of a region is an integer and defined in \cite[pp. 277]{alexander1928topological} 
such that the index of a given region is chosen at random, and the indices of the remaining regions 
are uniquely defined by the property that if a region with index $e$ is on the left of a strand, 
then the region on the right of that strand has index $e-1$.
\end{remark}

\begin{definition}\label{def:checkerboard}
The \textbf{checkerboard coloring} of a knot diagram is a coloring of the regions with two colors (white and black), in such a way that 
the regions with even index are colored white and the regions with odd index are colored black.
\end{definition}

\begin{remark}\label{rm:checkerboard}
In a checkerboard coloring of a knot diagram the two regions adjacent to a strand have distinct colors. 
So at each crossing, two regions have the same color if and only if they are not adjacent.
Notice that the checkerboard coloring does not depend on the orientation of the knot, 
but it does depend on the random value of of the index  of the given region in Definition \ref{r-index} 
in such a way that the colors black and white are interchanged when the random value is changed from even to odd and vice versa.
So a knot diagram has two checkerboard colorings which can be obtained from each other by interchanging the colors white and black.
\end{remark}

\begin{remark}\label{r-trivial-Dehn-col}
Let $x$ be the trivial Fox $(R,t)$-coloring where all the strands have the same color. Then for a fixed region $U_1$ with a fixed color, there exists Dehn $(R,t)$-coloring $U$ such that $\varphi (U)=x$ by Proposition \ref{prop:fox2dehn}. In this way we get all the so called \textbf{trivial} Dehn $(R,t)$-colorings which constitute a free summand $R^2$ of the module of Dehn $(R,t)$-colorings.
In particular, if $t\not=1$, $U_1=1$ and $x=0$, then all the regions $U_i$ have color $t^{e_i}$ where $e_i$ is the index of region $U_i$.
If $t=-1$, then the \textbf{trivial} Dehn $(R,t)$-colorings are such that all white regions have the same color and all the black regions have the same color.
A knot diagram is called \textbf{Dehn $(R,t)$-colorable} if it has a non-trivial coloring. Note that with these steps, trivial Fox colorings will transform into trivial Dehn colorings and the other way around, as well. Hence a knot diagram is Dehn $(R,t)$-colorable if and only it is Fox $(R,t)$-colorable.
\end{remark}

The next result allows us to compare some properties of the Alexander matrix and the Dehn coloring matrix of the same knot diagram.

\begin{proposition}\label{p-elem-id-dehn}
Let $N(t)$ be the Dehn $(R,t)$-coloring matrix of a knot diagram of a knot $K$,  
then $E_1(N(t))=0$ and $E_2(N(t))$ is a principal ideal generated by $\Delta_K(t)$.
\end{proposition}

\begin{proof}Choose two columns that correspond to two regions that have index $e$ and $e+1$ for some $e$. See Remark \ref{r-trivial-Dehn-col}.
Let $N''(t)$ be the $(n+2) \times n$ matrix  that is obtained from $N(t)$ by replacing the two chosen columns by zero columns.
Let $N_0(t)$ be the $n \times n$ matrix  that is obtained from $N(t)$ by deleting the two chosen columns.
The matrix $N(t)$ is equivalent to the matrix $N''(t)$. See \cite[pp. 280]{alexander1928topological}. 
So $E_1(N(t))=0$ and~$E_2(N(t))=E_2(N''(t))=E_0(N_0(t))$ by Proposition \ref{p-free-matrix}, which is a principal ideal generated by $\det (N_0(t))$.

Let $N_1(t)$ be the $(n+1) \times n$ matrix  that is obtained from $N(t)$ by deleting the chosen column of index $e$. In order to show that $E_2(N(t))$ is generated by $\Delta_K(t)$ we need to refer 
to the fact that the matrix $N_1(t)$ is obtained by the free calculus of 
the Dehn representation of the fundamental group of the complement of the knot by Remark \ref{r-dehn}, 
and $M(t)$ is the Alexander matrix obtained by the free calculus of 
another representation of the same fundamental group by Remark \ref{r-alex}. 
Therefore, these matrices are equivalent and have the same elementary ideals. 
See \cite[Chapter VII (4.5)]{crowell2012introduction}.
\end{proof}

We conclude this subsection with an example verifying that $E_1(M(t))$ and $E_2(N(t))$ are both generated by $\Delta_K(t)$.

\begin{example}
Consider the diagram of the trefoil knot whose Fox coloring matrix is given in Example \ref{ex:alexander}, and its Alexander polynomial is computed as $\Delta_K(t) = t^2-t+1$. Its Dehn coloring matrix is 
\begin{equation*}
N(t) = \begin{pmatrix}
1& -t & -1 & t & 0 \\
1 & -1 & 0 & t & -t \\
1 & 0 & -t & t & -1
\end{pmatrix}.
\end{equation*}

Let $N_{ij}^*(t)$ be the matrix obtained by deleting the $i$-th and the $j$-th column of $N(t)$ for~$i \neq j$. We have $M_{ij}^*(t) \in \{\pm (t^2 -t +1)\}$ and~$N_{ij}^*(t) \in \{ 0, \pm (t^3-t^2+t), \pm (t^2-t+1), t^3+1\}$. One can check that they are both generated by $\Delta_K(t) = t^2-t+1$, since~$t^3+1 = (t^2-t+1)(t+1).$
\end{example}

\subsection{Alexander-Briggs Coloring}
In this subsection, we study a third way to color knot diagrams, called Alexander-Briggs (AB) colorings. The \textit{AB coloring} is the coloring of the vertices of the Tait diagrams.
\begin{figure}[h!]
    \centering
    
\begin{tikzpicture}[pics/arrow/.style={code={%
  \draw[line width=0pt,{Computer Modern Rightarrow[line
  width=1pt,width=3ex,length=2ex]}-] (0.5ex,0) -- (-0.5ex,0);
  }},pics/rrarrow/.style={code={%
  \draw[line width=0pt,{Computer Modern Rightarrow[line
  width=1pt,width=3ex,length=2ex]}-] (-0.5ex,0) -- (0.5ex,0);
  }}]

\fill[orange!30] (0,2) -- (1,3) -- (-1,3)  -- cycle;
\fill[orange!30] (-7/3,-1/3) -- (-3,1) -- (-3,-1)  -- cycle;

\fill[orange!30] (-3.5,-2) -- (-3/2,-2) -- (-1,-3)  -- cycle;
\fill[orange!30] (0,-3) -- (3/4,-2) -- (3,-2) -- cycle;

\fill[orange!30] (15/7,-1/7) -- (3,1) -- (3,-1)  -- cycle;
\fill[orange!30] (0,2) -- (-7/3,-1/3) -- (-3/2,-2) -- (3/4,-2) -- (15/7,-1/7) -- cycle;

\begin{knot}[
clip width=0,
]
\strand[black,thick] (-3,-1) to pic[pos=0.45,sloped]{arrow} (1,3);
\strand[black,thick] (-1,3) to pic[pos=0.55,sloped]{arrow} (3,-1);
\strand[black,thick] (3,1) to pic[pos=0.5,sloped]{arrow} (0,-3);
\strand[black,thick] (3,-2) to pic[pos=0.5,sloped]{arrow} (-3.5,-2);
\strand[black,thick] (-1,-3) to pic[pos=0.5,sloped]{rrarrow} (-3,1);
\end{knot}

\tikzset{nnode/.style = {shape=circle,fill=myg,draw,inner sep=0pt,minimum
size=1.9em}}
\node[text=blue] at (-2.25,0.1) {$v_1$};
\node[text=blue] at (0.4,2) {$v_2$};
\node[text=blue] at (2.2,-0.6) {$v_3$};
\node[text=blue] at (0.9,-2.3) {$v_4$};
\node[text=blue] at (-1.7,-2.3) {$v_5$};

\node[text=red] at (-0.1,-0.3) {$U_{11}$};
\node[text=red] at (0,2.7) {$U_1$};
\node[text=red] at (1.5,1.5) {$U_2$};
\node[text=red] at (2.77,0) {$U_3$};
\node[text=red] at (2,-1.4) {$U_4$};
\node[text=red] at (1.5,-2.25) {$U_5$};
\node[text=red] at (-0.5,-2.5) {$U_6$};
\node[text=red] at (-2.2,-2.3) {$U_7$};
\node[text=red] at (-2.5,-1.5) {$U_8$};
\node[text=red] at (-2.77,-0) {$U_9$};
\node[text=red] at (-1.5,1.5) {$U_{10}$};

\node[text=black] at (0,2.3) {$\bullet$};
\node[text=black] at (-0.3,2) {$\bullet$};
\node[text=black] at (2.4,-0.15) {$\bullet$};
\node[text=black] at (2.1,0.15) {$\bullet$};
\node[text=black] at (0.4,-2.2) {$\bullet$};
\node[text=black] at (0.7,-1.8) {$\bullet$};
\node[text=black] at (-1.8,-1.8) {$\bullet$};
\node[text=black] at (-1.4,-1.8) {$\bullet$};
\node[text=black] at (-2.6,-0.3) {$\bullet$};
\node[text=black] at (-2.4,-0.6) {$\bullet$};

\end{tikzpicture}
    \caption{Tait diagram of a knot with a checkerboard coloring.}
    \label{fig:Tait}
\end{figure}
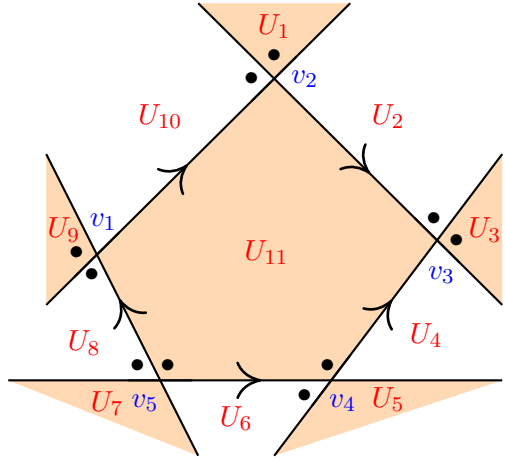

\begin{definition}
\label{def:Tait}
The \textbf{Tait diagram} of a knot is a diagram of that knot with a chosen orientation, with two additional dots at every crossing beside the left 
hand side of an overstrand such that one is placed just before and the other is placed just after the understrand, 
together with a chosen checkerboard coloring of the regions, see Figure \ref{fig:Tait}.
\end{definition}

The ``vertices" in knot diagrams are changed into ``crossings" in Tait diagrams which fits with the subsequent part of the paper where on the graph of a Tait diagram we have vertices and edges that are in fact the regions and the crossings, respectively, of the Tait diagram.

\begin{definition}
\label{def:AB-coloring}
Consider a Tait diagram of a knot. Define
$$
\wt (U) = \sum_{v \in \partial U} \wt (U,v) v
$$
where $U$ is a region of the diagram and $v$ is a vertex in the boundary $\partial U$ of $U$ and
$$
\wt (U,v) = 
\begin{cases}
t & \text{if there is a dot in } U \textnormal{ near } v,\\
1 & \text{otherwise}. 
\end{cases}
$$
is defined as the \textbf{weight} of $U$ at $v$. An \textbf{Alexander-Briggs (AB) $(R,t)$-coloring} is a coloring of the crossings with colors that are elements of $R$ 
in such a way that $\wt(U) = 0$ for all the regions $U$ of the the Tait diagram. A Tait diagram is called \textbf{Alexander-Briggs $(R,t)$-colorable} if it has a nonzero Alexander-Briggs $(R,t)$-coloring.  
\end{definition}

As an example, for an Alexander-Briggs $(R,t)$-coloring of the diagram in Figure \ref{fig:Tait} it is necessary that  $\wt(U_{11}) = 0$, that is,
$$
v_1 + v_2 + v_3 + tv_4 + tv_5 = 0,
$$
according to Definition \ref{def:AB-coloring}.

Consider a knot diagram with $n$ crossings. By Lemma \ref{lem:knot_diag}, we let $v_1,v_2, \ldots ,v_n$ be an enumeration of the crossings, and $x_1,x_2, \ldots ,x_n$ an enumeration of the strands and $U_1,U_2, \ldots , U_{n+2}$ an enumeration of the regions in the sequel. With a slight abuse of notation, we denote the colors assigned to these crossings, strands or regions with the same notation as their enumeration. Analogous to Definition \ref{def:DehnMat}, we define the following matrix for AB colorings.

\begin{definition}
\label{def:ABmatrix}
The \textbf{Alexander-Briggs $(R,t)$-coloring matrix} $P(t)$ of a Tait diagram with $n$ crossings is defined by $P_{rs}(t) = \wt(U_r,v_s)$ for $1 \le r \le n+2$ and $1 \le s \le n$ with $r,s \in \Z_{>0}.$
\end{definition}

Consider the morphism $R^n \rightarrow R^{n+2}$ of $R$-modules
given by the matrix $P(t)$. The \textbf{module of Alexander-Briggs $(R,t)$ colorings} of the Tait diagram is given by the kernel of this morphism, that is the solution space of the set of equations:
$$
\sum_{s=1}^n \wt(U_r,v_s) v_s \textnormal{ for } r=1,2,\ldots , n+2.
$$

\begin{proposition}
\label{p-PN}
Let $D_{\pm}$ be the $(n+2)\times (n+2)$ diagonal matrix with $1$ at entry $(i,i)$ if the region $U_i$ is white, and $-1$ if the region $U_i$ is black. Then $P(t)^T = N(t)D_{\pm}$ where the matrices $N(t)$ and $P(t)$ are as in \ref{def:DehnMat} and \ref{def:ABmatrix}, respectively.
\end{proposition}

\begin{proof}
The definition of $N(t)$ and $P(t)$ are such that the entries of $N_{ij}(t)$ and $P_{ji}(t)$ are the same up to a sign, 
and this sign is $+1$ if the region $U_j$ is white and $-1$ if the region $U_j$ is black.
\end{proof}

The next result shows that the module of AB colorings is invariant under Reidemeister moves.

\begin{corollary}
\label{cor:ReidemesterP}
If $D_1$ and $D_2$ are two equivalent Tait diagrams of knots, 
then the corresponding Alexander-Briggs $(R,t)$-coloring matrices~$P_1(t)$ and~$P_2(t)$ are equivalent.
\end{corollary}
\begin{proof}
This is a direct consequence of Propositions \ref{p-ReidemesterN} and \ref{p-PN}.
\end{proof}

\begin{corollary}
\label{cor:DAB}
Let $R$ be a field and $t$ a nonzero element of $R$. Then the dimension of the space of Dehn $(R,t)$-colorings is $2$ more than the dimension of the space of Alexander-Briggs $(R,t)$-colorings.
\end{corollary}

\begin{proof}
Let $D_{\pm}$ be as in \ref{p-PN} of size $n+2$. The ranks of $P(t)$ and $N(t)$ are the same by Proposition \ref{p-PN} since $D_{\pm}$ is an invertible matrix. The dimension of the module of Dehn $(R,t)$-colorings is equal to $n+2- rank(P(t))$.
The dimension of the module of Alexander-Briggs $(R,t)$-colorings is equal to $n- rank(N(t))$. 
\end{proof}

We conclude the section with a key remark that connects the three notions of colorability, showing that a Tait diagram is Alexander-Briggs $(R,t)$-colorable if and only if it is Dehn $(R,t)$-colorable if and only if it is Fox $(R,t)$-colorable. 

\begin{remark}\label{r-morphDtoAB}
One can generalize \cite[Theorem 3.1]{carter2014three} to show that 
there is a surjective morphism from the module of Dehn $(R,t)$-colorings to the module of Alexander-Briggs $(R,t)$-colorings 
that has as kernel a free $R$-module of rank $2$ consisting of the trivial Dehn $(R,t)$-colorings.
Hence a Tait diagram is Alexander-Briggs $(R,t)$-colorable if and only if it is  Dehn  $(R,t)$-colorable.
We saw already in Remark \ref{r-trivial-Dehn-col} that  a diagram is Fox $(R,t)$-colorable if and only if it is  Dehn  $(R,t)$-colorable.
Hence the three notions of colorability of a diagram coincide.
\end{remark}

\section{Codes from Knots and Their Properties}
\label{sec:codes}
\label{sec:3}

This section explains how 
one can construct a code
starting from 
a knot with its diagram and coloring. 
We also establish a series of results illustrating how the properties of knots determine those of codes via the said constructions.
We essentially regard
the three possible colorings of a knot diagram as a linear code over a finite field $\F_q$ with $q$ elements, that is $R = \F_q$. 
\begin{definition}
\label{def:knotcode}
Let $D$ be a knot diagram that is Fox $(\F_q,t)$-colored.
The \textbf{Fox code} associated with $D$ (or the \textbf{Fox knot code} of $D$) with  coloring matrix $M$ is 
$$
\mF_{D,t} = \{x \in \F_q^n \mid Mx^T = 0 \}.
$$
If $t=-1$, we denote this code by $\mF_{D}$. 
Similarly we define the \textbf{Dehn code}  and the \textbf{Alexander-Briggs code} of $D$ by
$$
\mD_{D,t} = \{x \in \F_q^n \mid Nx^T = 0 \} \ \mbox{ and } \ \mA_{D,t} = \{x \in \F_q^n \mid Px^T = 0 \}, 
$$
respectively where $N$ and $P$ are as in Definition \ref{def:DehnMat} and Definition \ref{def:ABmatrix}. The codes $\mF_{D,-1}$, $\mD_{D,-1}$ and $\mA_{D,-1}$ are abbreviated by $\mF_{D}$, $\mD_{D}$ and $\mA_{D}$, respectively.
\end{definition}

Note that the coloring matrix is interpreted as the parity check matrix of a code; see Definition \ref{def:parity}. 
A relabeling of the strands in the knot diagram will result in a possibly different code. 
Such a code is permutation equivalent to the original one 
and thus has the same dimension and minimum distance.
However, the Fox code of a knot diagram is not a knot invariant, as the following example illustrates.

\begin{example}
Let $q=19$. In Figure \ref{fig:76_flype} two diagrams of the same knot $K$ is given with~$\Delta_K(-1) = 19$. The Fox $(\F_{19},-1)$-coloring matrices of the knot diagrams depicted in Figure~\ref{fig:76_flype}, denoted by $H_a$ and $H_b$ respectively, are the parity check matrices of the corresponding knot codes.

\begin{figure}[H]
    \centering
    \includegraphics[width=\textwidth]{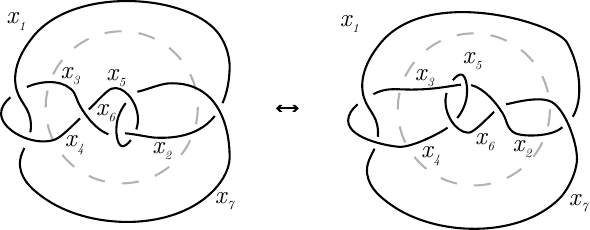}
    \caption{Two alternating diagrams of the $7_6$ knot}
    \label{fig:76_flype}
\end{figure}

We apply row operations to transform the parity check matrices in the form $\bigl[-A^T \mid I \bigr]$, for some matrix $A$ of suitable size. 
This results in the generator matrices 
$$
G_a = \begin{pmatrix}
1 & 0 & 6 & 15 & 16 & 3 & 10\\
0 & 1 & 14 & 5 & 4 & 17 & 10
\end{pmatrix}
$$
and 
$$
G_b = \begin{pmatrix}
1 & 0 & 6 & 15 & 3 & 9 & 10\\
0 & 1 & 14 & 5 & 17 & 11 & 10
\end{pmatrix}.
$$
It is easy to see that the codes generated by them are not monomial equivalent. Therefore, the Fox code of a knot diagram is not a knot invariant. 
\end{example}

In the theory of knot colorings one is interested in the (minimum) number of colors used in a coloring. This number cannot be translated in results about the weight of the coloring, that is, the number of nonzero colors. In 1999, Kauffman and Harary conjectured the following \cite{harary1999} and it was proven in 2009 by Mattman and Solis in \cite{mattman2009proof}.

\begin{theorem}
\label{thm:conjecture}
Let $D$ be a reduced, alternating knot diagram (see Definition \ref{def:reduced_alt}) of $K$ with~$|\Delta_K(-1)|=p$, where $p$ is prime. Then, every non-trivial Fox $(\F_p,-1)$-coloring of $D$ assigns different colors to different strands of the diagram.
\end{theorem}

Note that Theorem \ref{thm:conjecture} is not true if the determinant of the knot is not a prime. The alternating knot diagram $7_7$ in \cite[Figure 25]{harary1999} has non-prime determinant $21$ and has a Fox $(\F_7,-1)$-coloring with $6$ colors such that two strands have the same color.

The Kauffman-Harary conjecture of Theorem \ref{thm:conjecture} motivates the following result.

\begin{proposition}
Let $D$ be a reduced, alternating knot diagram (see Definition \ref{def:reduced_alt}) of~$K$ with $n$ strands such that $|\Delta_K(-1)|=p$, with $p$ prime. Then, the Fox knot code of $D$ is an $[n,2,n-1]_p$ code over $\F_p$.
\end{proposition}

\begin{proof}
By Theorem \ref{thm:conjecture}, every non-trivial coloring of $D$ assigns different colors to different strands. This implies that the minimum distance of the Fox knot code is $n-1$ which is attained by~$d(c,c')$ where $c$ is any non-trivial coloring and $c'$ is a trivial coloring where all strands have color $c_i$ for some $i \in \{1,\ldots,n\}$.  By Proposition \ref{prop:determinant}, we have that $D$ is Fox $(\F_p,-1)$-colorable since $p \mid \Delta_K(-1)=p$. Non-trivial colorability implies that the dimension of the Fox knot code is at least $2$. The only possible code parameters are $[n,2,n-1]_p$ by the Singleton bound of Theorem \ref{thm:sbound}.
\end{proof}

Note that the Fox knot codes with Fox coloring matrix as their parity check matrices are right $3$-regular LDPC codes. If a knot diagram is alternating, it gives a $(3,3)$-doubly-regular LDPC code. Moreover, if one considers the Dehn colorings, then the corresponding code is a right $4$-regular LDPC code. For the rest of this section, when we say coloring matrix, we mean the Alexander matrix of Definition \ref{def:coloringMat}.

Regarding the minimum distance of a Fox knot code, one can obtain the following, rather simple, result.

\begin{proposition}
\label{prop:minimumdistance}
A Fox code of a knot diagram of a non-trivial knot has minimum distance at least 2.
\end{proposition}
\begin{proof}
Suppose there exists a Fox code of a knot diagram with minimum distance 1. Then this code contains a codeword of weight 1, which corresponds to a coloring of the knot diagram in which only one of the strands is colored with a color $c \in R\backslash\{0\}$. In case this strand is an overstrand at a crossing in the diagram, it is also an understrand at another crossing in the diagram, unless it is the trivial knot. Then, there exists a crossing for which it should hold that $0 - 0 = t(c - 0)$ or $c - 0 = t(0 - 0)$, depending on which understrand of the crossing is colored. It follows that $c = 0$ should hold as $t$ is invertible over $R$. From this contradiction, it follows that the minimum distance of the code is at least 2.
\end{proof}

We will return to the minimum distance of Fox knot codes in Remark~\ref{rem:minimumdist} and in Subsection~\ref{subsec:pretzel}. We can already disclose that it is not a knot invariant; see Remark \ref{rem:minimumdist} for the details.

\subsection{Dimension of Codes from Knot Diagrams}
\label{subsec:dim}

In this subsection we investigate
the dimension of a Fox code of a knot diagram. In particular, we prove that the dimension of a Fox code of a knot diagram is a knot invariant. We start with an observation.

\begin{remark}
\label{rem:dim1}
The $n$-repetition code of Example \ref{ex:repetition} is always a subcode of the Fox code of a knot diagram with $n$ strands, as trivial colorings are always possible. Therefore, the dimension of the Fox code of a knot diagram is at least 1. Conversely, when the dimension of the Fox code of a knot diagram is larger than 1, the knot has a non-trivial coloring.
\end{remark}

We directly start with one of the main theorems of the subsection.

\begin{figure}[ht]
    \centering
    \includegraphics[width=0.9\textwidth]{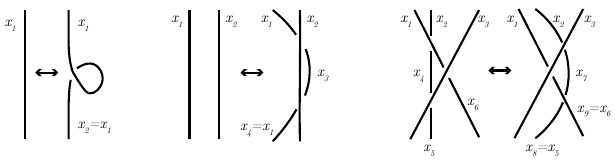}
    \caption{The effect of the Reidemeister moves on a Fox-coloring.}
    \label{fig:fox_reidemeister}
\end{figure}

\begin{theorem}
\label{thm:dimreidemeister}
Let $D$ and $D'$ be equivalent knot diagrams. Then 
$\mF_{D,t}$ and $\mF_{D',t}$ have the same dimension.
\end{theorem}

\begin{proof}
Let $D$ and $D'$ have $n$ strands. Denote by $\mF_{D,t}$ and $\mF_{D',t}$ the respective codes related to the diagrams. One locally investigate what happens when performing each Reidemeister move, see 
Figure~\ref{fig:fox_reidemeister}. Suppose $D'$ is obtained from $D$ by twisting a strand $x_1$ (Reidemeister move of type I), then the twist results in two strands and a crossing in this part of the diagram, where both strands are the understrands and one of the strands is the overstrand. For a Fox coloring it then follows that the colors assigned to both strands must be the same. 
    Let 
    $$G = 
    \begin{pmatrix}
     | & | & | \\
    \dots & x_1 & \dots \\
    | & | & |
    \end{pmatrix}
    $$ 
    be a full rank $k \times n$ generator matrix of $\mF_{D,t}$.
    Then    $$G' = 
    \begin{pmatrix}
    | & | & | & |\\
    \dots & x_1 & x_1 & \dots\\
    | & | & | & | 
    \end{pmatrix}
    $$ 
    is a $k \times (n+1)$ generator matrix for $\mF_{D'}$, which has the same rank since the added column is a duplicate of another column.

    The other moves can be investigated
    in a similar manner and we omit the proof here. By Theorem \ref{thm:reidemeistermoves},
any two diagrams of a knot can be transformed into each other using Reidemeister moves. It follows that $\mF_{D,t}$ and $\mF_{D',t}$ have the same dimension.
\end{proof}

The next result derives an upper bound for the dimension of Fox knot codes. 

\begin{theorem}
\label{thm:dimension-ineq}
Let $D$ be a knot diagram with $n$ strands and let $\mF_{D,t}$ be the corresponding Fox knot code over $\F_q$. We have
$$
1 \le \dim(\mF_{D,t}) \le \frac{n+1}{2} .
$$ 
\end{theorem}
\begin{proof}
By Theorem \ref{thm:dimreidemeister}, the Reidemeister moves do not affect the dimension of a Fox code of a knot diagram. Therefore, without loss of generality, let $D$ be a knot diagram that does not have any crossings which can be removed via the twist move, see Definition \ref{def:reidemeister} 
(that is, there exists no crossing in $D$ such that the overstrand and one of the understrands of the crossing are the same strand).  
Label the strands of $D$ as $x_1, x_2, ..., x_n$ by following the knot in one direction. Denote by $c_i$ the crossing where $x_i$ and $x_{i+1}$ are understrands and $x_{i'}$ is the overstrand, and with $x_{i+1}$ being the strand leaving the crossing with respect to the chosen orientation. Note that $i' \in \{1, \ldots, n\}$ depends on $i$. Since $D$ does not contain twists, we have that either $i' < i$ or $i' > i+1$. By going through the crossings $c_i$ with $1 \le i \le n-1$, we define the sets $L = \{ c_i \mid 1 \leq i \le n-1,\  i' < i\}$ and~$U= \{ c_i \mid 1 \leq i \le n-1,\  i' > i+1\}$. As $L \cap U = \emptyset$ and $L \cup U = \{1, \ldots, n-1\},$ we find that either $|L| \geq \frac{n-1}{2}$ or $|U| \geq \frac{n-1}{2}$, since $|L \cup U| = n-1.$ 
Then, the parity check matrix $H$ of $\mF_{D,t}$ can be constructed like in Definition~\ref{def:coloringMat} as follows. 

Let 
$$
H_{ij}(t) = 
\begin{cases}
x & \text{if} \ \ j=i,\\
y & \text{if} \ \ j \equiv i+1 \pmod n,\\
1-t & \text{if} \ \ j=i',\\
0 & \text{otherwise},
\end{cases}
$$ where $(x,y) \in \{(-1,t),(t,-1)\}$ depending on the diagram as in Definition~\ref{def:coloringMat}. 

Throughout the rest of the proof, we show that $\rk(H) \ge \frac{n-1}{2}$. When~$|U| \geq \frac{n-1}{2}$, take the submatrix $H'$ of $H$ consisting of the rows corresponding to the crossings in $U$. 
Then $H'$ is in row echolon form since the $i$-th row of $H$ such that $c_i \in U$ only has nonzero entries at positions $i$, $i + 1$ and $i'$ with $i' > i+1$. Thus, $\rk(H') = |U|$. This implies that $\rk(H) \ge |U|$ and consequently $\mF_{D,t}$ has dimension at most $n - |U|$. Similarly, if $|L| \geq \frac{n-1}{2}$ then the submatrix $H''$ consisting of the rows of $H$ corresponding to the crossings in $L$ is in column echelon form since the $i$-th row of $H$ such that $c_i \in L$ only has nonzero entries at positions $i$, $i + 1$ and $i'$ with $i' < i$. Therefore, we have $\dim(\mF_{D,t}) \leq n - \frac{n-1}{2} = \frac{n+1}{2}$.
\end{proof}

In addition to Theorem \ref{thm:dimension-ineq}, the following two results hold 
about the dimension of Fox knot codes.

\begin{proposition}
\label{prop:dimension}
Let $D$ be a knot diagram with $n$ strands and let $\mF_{D,t}$ be the corresponding Fox knot code over $\F_q$. 
Let $\dim(\mF_{D,t}) = k$. Then $k$ is the smallest integer with the property that~$E_k(M(t))=\F_q$.
\end{proposition}

\begin{proof}
Let $\dim(\mF_{D,t}) = k$. Then $M(t)$ has rank $r=n-k$, since $\mF_{D,t}$ is the null space of the matrix $M(t)$. 
Proposition \ref{p-elem} states that $E_l(M(t))=\F_q$ if $l\geq n-r=k$ and $E_l(M(t))=0$ if $l<k$.
Hence $k$ is the smallest integer such that $E_k(M(t))=\F_q$.
\end{proof}

Note that Proposition~\ref{prop:dimension} is also stated in \cite[Corollary 12]{traldi2018}, where it has a longer proof. 

\begin{proposition}\label{p-dim-prime-e}
Let $p$ be a prime number and let $t$ an integer such that $1\leq t<p$.
Let~$D$ be a knot diagram and let
$\mF_{D,t}$ be the corresponding Fox knot code over $\F_p$. 
Let $e$ the largest integer such that $p^e$ divides $\Delta_K(t)$ in $\Z$. 
Then $\dim(\mF_{D,t}) \leq e+1$, and equality holds if $e=1$.
\end{proposition}

\begin{proof}
Our proof uses Proposition \ref{p-number-color-Z} with $R=\Z_t$, the localization of $\Z$ at $t$, and $d=p$. In particular,  $\overline{R}=\F_p$ and 
$t$ is an invertible element in $\F_p$.
Let $(d_1)\subseteq (d_2)\subseteq \cdots \subseteq (d_l)$ be the invariant factors of the module of Fox $(\Z_t,t)$-colorings.
Then $\Delta_K(t)=\prod_{i=2}^n d_i$ and $d_1=0$ by Proposition \ref{p-elem-ideals-Alex-pol}.
Let $\overline{d_i}=\gcd (p,d_i)$ and let $\overline{e}$ be the number of integers $i$ with $2 \leq i \leq l$ and $\overline{d_i}=p$. 
Furthermore, $p^{\overline{e}}$ divides $\Delta_K(t)$ in $\Z$, hence
$\overline{e} \leq e$.
Then $p\prod_{i=2}^n\overline{d_i}$ is the number of Fox~$(\F_p,t)$-colorings of $D$ by Proposition~\ref{p-number-color-Z}. So $\dim(\mF_{D,t})=1 + \overline{e}$
and $\dim(\mF_{D,t})\leq 1 + e$. If~$e=1$, then the underlying knot, say $K$, is Fox $(\overline{R},t)$-colorable by Proposition \ref{prop:determinant}. 
So~$1<\dim(\mF_{D,t})\leq 1 + e$
and $\dim(\mF_{D,t})=2$.
\end{proof}

We finish this section with the following observation.

\begin{remark}
\label{rem:minimumdist}
The minimum distance of codes of knots is not a knot invariant. This can be seen, for example, from the generator matrix of the first Reidemeister move in the proof of Theorem~\ref{thm:dimreidemeister}.
\end{remark}

\section{Two Families of Fox Knot Codes}
\label{sec:families}
\label{sec:4}

This section is devoted to the study of two families of knots and their codes, namely \textit{torus knots around other knots} and \textit{pretzel knots}.
These can be both used to construct codes with 
interesting parameters and will be treated in dedicated subsections. In the sequel, for ease of notation we will write
$E_k(t)$
instead of
$E_k(M(t))$,
where $M(t)$ is the coloring matrix of the knot diagram at hand.

\subsection{Torus Knots}
\label{subsec:torus}
The notion of a torus knot was already introduced in Definition \ref{def:torusknot}. In this subsection we investigate the properties of these knots and their generalizations. We then study the dimension of codes of knot diagrams of these knots and show how to construct codes of arbitrary dimension.

\begin{remark}\label{tubular-nbhd}
There exits a closed \textbf{tubular neighbourhood} of $K$, denoted by $\tub(K)$, 
such that $\tub(K)$ is homeomorphic to $S^1\times D^2$ via a homeomorphism $h$ where $K$ is mapped to $S^1\times \{0\}$,
and the boundary of $\tub(K)$ is homeomorphic to the torus $S^1\times S^1$. See \cite{hirsch1968} for more details.
\end{remark}

Generalizations of torus knots are defined as follows; see \cite{burau1934knots,le1972knots}.

\begin{definition}
\label{def:cableknot}
Let $\tub(K)$ be a {tubular neighbourhood} of~$K$.
Let~$a,b$ be positive integers that are relatively prime.
Then the curve on $S^1\times S^1$ given by the parametrization $\varphi (t)=at$, $\theta (t)=bt$, is mapped via $h^{-1}$ of Remark \ref{tubular-nbhd} to a knot on the
boundary of the tubular neighbourhood $\tub(K)$. 
This knot is called the $(a,b)$-\textbf{torus knot around} $K$ and is denoted by $K(a,b)$. 
By induction, we can repeat this procedure for some integer $m \in \Z_{\ge 2}$ to obtain the $(a_1,b_1,\ldots ,a_m,b_m)$-\textbf{iterated torus knot} $K(a_1,b_1,\ldots ,a_m,b_m)$ \textbf{around} $K$, where
the pairs $(a_i,b_i)$ are relatively prime  and $K(a_1,b_1,\ldots ,a_i,b_i)$ is the $(a_i,b_i)$-torus knot around
$K(a_1,b_1,\ldots ,a_{i-1}, b_{i-1})$ for all~$i \in \{1,\ldots,m\}$. 
\end{definition}

\begin{remark}
    It can be seen that Definition~\ref{def:cableknot} generalizes torus knots, in the sense that the torus knot $T(a,b)$ is the $(a,b)$-torus knot around the trivial knot or unknot $U$.
\end{remark}

The Alexander polynomial of a torus knot has a rather simple expression.
 
\begin{proposition}\label{p-alex-torus}

Let~$a,b$ be positive integers that are relatively prime.
The Alexander polynomial of the torus knot $T(a,b)$ is given by 
$$
\Delta_{T(a,b)}(t) =\frac{(t^{ab}-1)(t-1)}{(t^a-1)(t^b-1)}.
$$
Moreover, the $k$-th elementary ideal of $T(a,b)$ is $\Z[t,t^{-1}]$ for all $k \geq 2$.
\end{proposition}
\begin{proof}
See  \cite[Chapter 3]{burau1934knots}, \cite[VIII Exercise 3]{crowell2012introduction}, 
and \cite[Theorem 7.3.2]{murasugi1996knot}.
\end{proof}

We give an example to show how Proposition \ref{p-alex-torus} can be used.

\begin{example}\label{ex:torus}

Let~$a,b$ be positive integers that are relatively prime. Let $D$ be a diagram of $T(a,b)$. Proposition \ref{prop:determinant} implies the following.
\begin{enumerate}
    \item If $a$ and $b$ are odd, then $\Delta_{T(a,b)}(-1)=1$ and there are only trivial Fox $(\F_p,-1)$-colorings of $D$.
    \item If $a$ is odd and $b$ is even, then $\Delta_{T(a,b)}(-1)=a$, and $D$ is Fox $(\F_p,-1)$-colorable if and only if $p$ divides $a$.
    \item If $b$ is odd and $a$ is even, then $\Delta_{T(a,b)}(-1)=b$, and $D$ is Fox $(\F_p,-1)$-colorable if and only if $p$ divides $b$.
\end{enumerate}

The dimension over $\F_p$ of $\mF_D$ is 1 in the first case and 2 in the second and the third case by Proposition \ref{prop:dimension}, since the second elementary ideal is the whole ring by Proposition \ref{p-alex-torus}. If $ab$ divides $q-1$, then there exists an element $t$ in $\F_q^*$ of order $ab$. So $\Delta_{T(a,b)}(t)=0$, $D$ is $(\F_q,t)$-colorable, and the dimension over $\F_q$ of $\mF_{D,t}$ is 2.
\end{example}

\begin{remark}\label{r-p-order-det} 
The inequality in Proposition \ref{p-dim-prime-e} is in general not an equality (see \cite[Chapter 3, §4, Exercise 4.6]{livingston1993knot}), 
contrary to what is stated in \cite[Theorem 23]{henrich2022}. This can also be seen by taking $K=T(2,9)$. We have~$\Delta_K(-1)=9$, and thus the largest integer $e$ such that $3^e$ divides 9 is $e=2$ in this case. However, the dimension of the code over $\F_3$ is equal to $2$, showing that the bound of Proposition \ref{p-dim-prime-e} is not sharp in general. 
\end{remark}

We can determine the elementary ideals of the knot $K(a,b)$ in terms of the elementary ideals of the knot $K$ and the Alexander polynomial of the torus knot $T(a,b)$.

\begin{proposition}\label{p-elem-torus-around-K}
Let~$a,b$ be nonzero integers that are relatively prime. We have
$$
\tilde{E}_k(t) = \Delta_{T(a,b)}(t)E_k(t^b) + E_{k-1}(t^b).
$$
where $E_k(t)$ denotes the $k$-th elementary ideal of $K$ and $\tilde{E}_k(t)$ denotes the $k$-th elementary ideal  of the knot~$K(a,b)$.
\end{proposition}
\begin{proof}
See \cite{le1972knots} and \cite[Proposition 10.5]{pellikaan1981knopen}.
\end{proof}

We have the following two corollaries of Proposition \ref{p-elem-torus-around-K}.

\begin{corollary}\label{c-alex-torus-around-K}
Let~$a,b$ be positive integers that are relatively prime. Then the Alexander polynomial of the $(a,b)$-torus knot around $K$ is given by
$$
\Delta_{K(a,b)}(t)=\Delta_{T(a,b)}(t)\Delta_K(t^b).
$$
\end{corollary}
\begin{proof}
$\Delta_{K(a,b)}(t)$ is a generator of the principal ideal $\tilde{E}_1(t)$,
$\Delta_K(t)$ is a generator of the principal ideal $E_1(t)$ in Proposition \ref{p-elem-torus-around-K}, and 
$\tilde{E}_1(t) = \Delta_{T(a,b)}(t)E_1(t^b) + E_0(t^b)$. This gives the desired result since $E_0(t) =(0)$.
\end{proof}

\begin{corollary}\label{c-dim-torus-around-K}
Let~$a,b$ be nonzero integers that are relatively prime such that $a$ is even and $b$ is odd. Let $p$ be a prime divisor of $b$.
Let $K$ be a knot, $D$ a diagram of $K$, and let~$k$ denote the dimension of $\mF_{D,t}$ over $\F_p$.
Then the code $\mF_{\tilde{D},t}$ of a diagram $\tilde{D}$ of $K(a, b)$ has dimension $k + 1$ over $\F_p$.
\end{corollary}
\begin{proof}
Let $E_l(t)$ be the $l$-th elementary ideal of $K$. Then  $E_l(-1)=\F_p$ if $l\geq k$, and $E_l(-1)=0$ if $l<k$ by Proposition \ref{prop:dimension}.
Since $p$ is a prime that divides~$b$, $a$ is even and~$b$ is odd, we have $\Delta_{T(a,b)}(-1)=0$ in $\F_p$ by Example \ref{ex:torus} and $(-1)^b=-1$.
Let $\tilde{E}_l(t)$ be the $l$-th elementary ideal of $(a,b)$-torus knot around~$K$.
Then  $\tilde{E}_l(-1)=\F_p$ if $l\geq k+1$, and $\tilde{E}_l(-1)=0$ if $l<k+1$, by Proposition \ref{p-elem-torus-around-K}.
Hence~$\mF_{\tilde{D}}$ has dimension $k+1$ by Proposition~\ref{prop:dimension}, as claimed.
\end{proof}

In the next example we show how to build codes using 
iterated torus knots around other knots.

\begin{example}
Let $p$ be an odd prime. A diagram of the iterated torus knot $K(2,p,\ldots ,2,p)$, where $K$ is the unknot and $(2,p,\ldots ,2,p)$ is the $m$-fold
repetition of $(2,p)$, gives a code over~$\F_p$ of dimension $m+1$ by Corollary \ref{c-dim-torus-around-K}. 
The recursive formula of the length of the code of~$K(2,p,\ldots ,2,p)$ is given by $n_1=3$, $n_{m+1} =4n_m+p$.
\end{example}

We conclude this subsection with the following crucial example.

\begin{example}\label{ex:min-dist-fox-dehn}
Let $b=2l+1$ be a positive odd integer for some $l$ and let $T(2,b)$ be the torus knot as given in Definition \ref{def:torusknot}. Consider its diagram depicted in Figure \ref{fig:2btorus}. This is a diagram with $b$ crossings where the upper left understrand is connected with the lower left overstrand, and the upper right overstrand is connected with the lower right understrand.
Denote the upper left understrand by $x_1$ and the upper right overstrand by $y_1$. 
Denote the strands by following the knot's orientation from 
the upper left understrand $x_1$ to the lower right understrand by~$x_1, x_2, \ldots ,x_{l+2}$, respectively.
Denote the strands following the knot's orientation from the
upper right overstrand $y_1$ to the lower left overstrand by $y_1, y_2, \ldots ,y_{l+1}$, respectively.
Then~$x_{l+2}=y_1$ and $y_{l+1}=x_1$, see again Figure \ref{fig:2btorus}.

\begin{figure}[h]
   \centering
   \begin{subfigure}{.4\textwidth}
        \centering
       \subcaptionbox{$T(2,2l+1)$.\label{fig:2btorus}}[.6\linewidth]
       {\includegraphics[width=3cm]{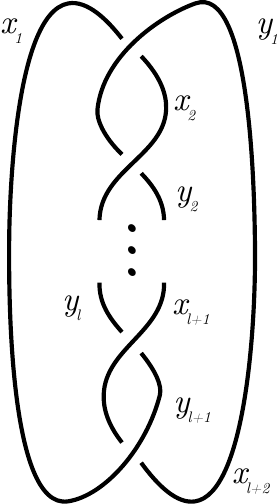}}
   \end{subfigure}
   \hspace{0.1\textwidth}
   \begin{subfigure}{.4\textwidth}
        \centering
        \subcaptionbox{A checkerboard coloring of $T(2,5)$.\label{fig:toruscheckerboard}}[.7\linewidth]
       {\includegraphics[width=2cm]{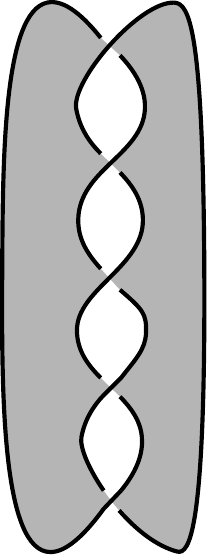}}
   \end{subfigure}
\caption{Diagram of $T(2,b)$ with $b$ is odd given in Example \ref{ex:min-dist-fox-dehn}.} 
\end{figure}

Let $p$ be a prime and suppose that the strands are Fox $(\F_p,-1)$ colored. 
Then we get by induction that $x_{i+1}=2iy_1-(2i-1)x_1$ and $y_{i+1}=(2i+1)y_1-2ix_1$.
So $x_{l+2}=y_1$ and~$y_{l+1}=x_1$ imply $by_1-bx_1=0$ in $\F_p$.
Hence $T(2,b)$ is  Fox $(\F_p,-1)$-colorable  if and only if $b$ is divisible by $p$.

If $b=p$, then we get a non-trivial coloring with~$x_{i+1}=2i$ and $y_{i+1}=2i+1$.
So all the strands have mutually distinct colors, which is in agreement with Theorem \ref{thm:conjecture}, since $\Delta_{T(2,p)} (-1)= p$ by Proposition \ref{p-alex-torus}. Furthermore, the Fox colorings have weight $1$ (all strands have color $0$), or $p$ (all strands have the same nonzero color), or $p-1$ for a non-trivial coloring, when the colorings are viewed as codewords as in Section \ref{sec:codes}. A checkerboard coloring (see Definition~\ref{def:checkerboard}) of the regions has~$p$ regions with color~$0$ (white in Figure \ref{fig:toruscheckerboard}), where the unbounded region is colored white, and two regions of nonzero color~(black in the figure). This gives a Dehn $(\F_p,-1)$-coloring of the diagram  of weight $2$. So the isomorphism of modules as mentioned in Proposition \ref{prop:fox2dehn} sends a word of weight $2$ to a word of weight $p$. Hence the isomorphism is not an isometry if $p>3$. 
\end{example}

\subsection{Pretzel Knots}
\label{subsec:pretzel}

In this subsection we prove that also pretzel knots codes can be used to construct codes with prescribed dimension. Moreover, we study the error correction capability of these codes. Starting from knots, one can create larger objects called \textit{links}.

\begin{definition}
\label{def:link}
Let $n \in \Z_{\ge 1}$. A \textbf{link} $L = \{K_1,\ldots, K_n\}$ is a finite collection of knots such that $K_i \cap K_j = \emptyset$ for all  $i,j \in \{1,\ldots,n\}$ with $i \neq j$. Each of the constituent knots is a \textbf{component} of the link. In particular, a \textbf{polygonal link} is a link each of whose component is a polygonal knot. 
\end{definition}

Since we only consider polygonal knots in this paper, we only consider polygonal links and simply write \textbf{link} for those.

\begin{figure}[h]
   \centering
    \includegraphics[width=3cm]{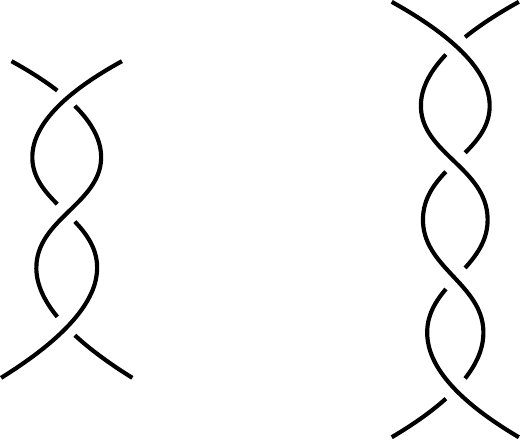}
    \caption{A (3)-crossing twist and a (-4)-crossing twist.}
    \label{fig:twists}
\end{figure}

Next, we informally define pretzel links following \cite{kawauchi1996survey}.

\begin{definition}
\label{def:pretzelink}
A \textbf{twist} is a part of a knot diagram consisting of two strands and at least a crossing such that all the crossings are obtained using both strands together, and it is of the form depicted in in Figure~\ref{fig:twists}. A twist with $|b| \in \Z_{>0}$ crossings is called a $(|b|)$-crossing twist if the top right strand is an overstrand, and is called a $(-|b|)$-crossing twist if the top right strand is an understrand. Let $p_1,\ldots,p_m$ be nonzero integers for some $m \in Z_{> 0}$. A \textbf{pretzel link} is a link with its diagram depicted as in Figure~\ref{fig:pretzellink}, where each rectangle denotes a twist with $|p_i|$ crossings. We denote this object by $P(p_1,\ldots,p_m)$. It is obtained when multiple twists are placed next to each other, where for each pair of neighboring strands the top and bottom right strands of the left twist are connected to the top and bottom left strands of the right twist, respectively, and the the top and bottom left strands of the leftmost twist are connected to the top and bottom right strands of the rightmost twist, respectively. 
\end{definition}

\begin{figure}[ht]
    \centering
    \includegraphics[width=5cm]{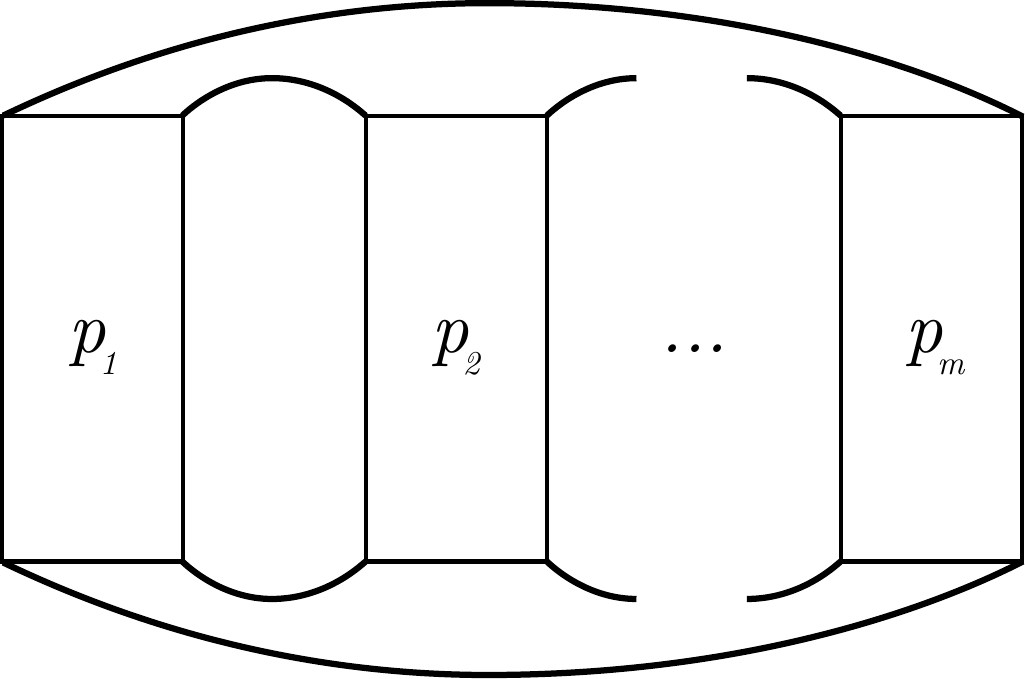}
    \caption{A general diagram of a pretzel link.}
    \label{fig:pretzellink}
\end{figure}

The sufficient and necessary condition when a pretzel link is a knot is proven in \cite{kawauchi1996survey}. 
\begin{proposition}
\label{lem:knot}
A pretzel link $P(p_1,\ldots,p_m)$ is a knot if and only if $m$ and $p_i$ are odd integers for all $i \in \{1,\ldots,m\}$, 
or $m \ge 1$ and exactly one of the the $p_i$ is even.
\end{proposition}

Next, we show that codes with any dimension can be constructed using pretzel knots.

\begin{theorem}(see \text{\cite[Theorem 17]{kolay2023}})
\label{thm:pretzeldim}
Let $D$ be a diagram of the pretzel knot $P(p_1, p_2, ..., p_m)$.
Let $q$ be a power of the prime $p$.
\begin{enumerate}
    \item If $p_i$ is coprime with $q$ for all $i \in \{1,\ldots,m\}$, 
then the dimension of a Fox knot code $\mF_D$ over~$\F_q$ is given by
$$
\dim(\mF_D) = 
\begin{cases}
2 & \text{if} \ \ p \mid \Delta_K(-1),\\
1 & \text{otherwise}.
\end{cases}
$$
\item If there exists a~$p_i$ that is not coprime with $q$ for some $i \in \{1,...,m\}$, then the dimension of $\mF_D$ over $\F_q$ is 
$|\{i \mid \gcd(p_i, q) \neq 1,\, i \in \{1, \ldots, m\} \}|$.
\end{enumerate}

\end{theorem}

We give an example to show an application of Theorem \ref{thm:pretzeldim}.

\begin{figure}[H]
    \centering
    \includegraphics[width=9cm]{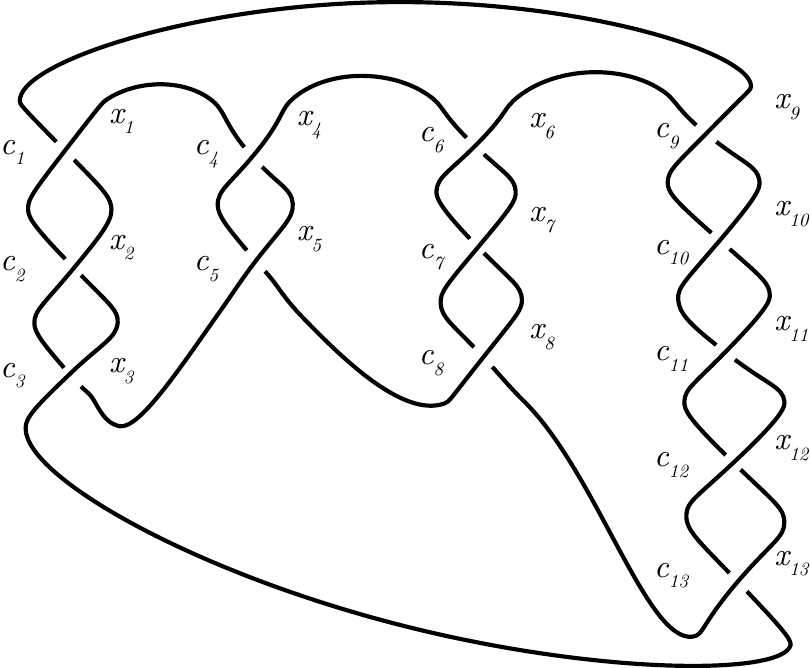}
    \caption{Diagram of the $P(3,2,3,5)$ pretzel knot.}
    \label{fig:3235pretzel}
\end{figure}

\begin{example}
\label{ex:pretzel}
The $P(3,2,3,5)$ pretzel knot is depicted in Figure \ref{fig:3235pretzel} has determinant $123 = 3 \cdot 41$, which means it is non-trivially colorable over $\mathbb{F}_3$ and $\mathbb{F}_{41}$ by Proposition \ref{prop:determinant}. By the second part of Theorem~\ref{thm:pretzeldim}, we then find that the code obtained from the colorings over $\mathbb{F}_3$ has dimension~2. By the first part of Theorem~\ref{thm:pretzeldim}, we also have that the code obtained from the colorings over $\mathbb{F}_{41}$ has dimension 2, as well.
\end{example}

For some special family of pretzel knots, we can determine the exact code parameters of the corresponding code. Computing the minimum distance in the general case seems to be a hard task.

\begin{proposition}
\label{prop:pretzmind}
Let $p$ be an odd prime and $D$ be a diagram of $P(p_1,\ldots ,p_m)$ with $p_i=p$ for all $i \in \{1,\ldots,m\}$. Then,  
$\mF_D$ is a $[pm,m,2p-2]_p$ code.
\end{proposition}

\begin{proof}
The statement about the dimension is already shown in Theorem \ref{thm:pretzeldim}.
Consider the numbering of the strands of the diagram $D$ analogous to Example \ref{ex:min-dist-fox-dehn}.
Let $p=2l+1$ for some $l \in \Z$.
For $i \in \{1,\ldots,m\}$, number the strands in the $i$-th block as follows. 
Following the knot's orientation, label the strands from the upper left $x_{i,1}$ to the lower right by $x_{i,1}, x_{i,2}, \ldots ,x_{i,l+2}$.
Following the knot's orientation, label
the strands from the
upper right $y_{i,1}$ to the lower left by $y_{i,1}, y_{i,2}, \ldots ,y_{i,l+1}$.
The strand $y_{i,1}$ is equal to $x_{i+1,1}$ for all~$i \in \{1,\ldots,m-1\}$, and $y_{m,1}$ is equal to $x_{1,1}$.
The strand $y_{i,l+1}$ is equal to $x_{i-1,l+2}$ for all~$i \in \{2,\ldots,m\}$, and $y_{1,l+1}$ is equal to $x_{m,l+2}$.

The $i$-th block consist of $2l+1=p$ strands
$\smash{x_{i,1}, x_{i,2}, \ldots ,x_{i,l+2}}$ and $\smash{y_{i,2}, \ldots ,y_{i,l}}$, since $\smash{y_{i,1}=x_{i+1,1}}$ and $\smash{y_{i,l+1}=x_{i-1,l+2}}$.
The values of $\smash{x_{1,1}, x_{2,1}, \ldots ,x_{m,1}}$ determine all the values of the other strands due the our assigning above.
Now $x_{i,1}=0$ and $x_{i+1,1}=0$ if and only if all the strands in the~$i$-th block have value zero.
If $x_{i,1}\not=0$, then at least $p-1$ strands of the $i$-th block and~$p-1$ strands of the $(i-1)$-th block have nonzero value as shown in Example~\ref{ex:min-dist-fox-dehn}.
Hence the weight of a nonzero codeword is at least $2(p-1)$.

Choosing $x_{1,1}=1$ and $x_{i,1}=0$ for all $i\neq 1$ gives a codeword of weight $2(p-1)$.
Hence~$\mF_D$ has indeed minimum distance $2(p-1)$ and rate $R=m/pm= 1/p$.
\end{proof}

\section{Knot Graphs and Their Codes}
\label{sec:graph}

Starting from Tait diagram of knots one can construct graphs, see \cite{kauffman1983,kauffman1987,kauffman1988,kauffman1989,kauffman1991,kauffman2012}. We assume that the reader is familiar with basic concepts in graph theory, see \cite{west2001introduction} as a reference. 

\begin{definition}
\label{def:signedgraph}
Let $D$ be a Tait diagram of a knot, and $D^*$ be equal to $D$ with the interchanged checkerboard coloring. The \textbf{black graph} of $D$ is the planar graph $\Gamma _D$ whose vertices are the black regions of $D$. There is an edge between two vertices if the black regions in the Tait diagram corresponding to these vertices have a crossing in their common boundaries. Similarly, $\Gamma _{D^*}$ is called the \textbf{white graph} of $D$. The graphs can be made directed by choosing the direction from the region without a dot to the region that has a dot near the crossing in their common boundary. See Figure \ref{fig:signedgraph} for illustration.
\end{definition}

\begin{figure}[ht!]
    \centering
\subcaptionbox{Tait Diagram of Figure \ref{fig:trefoil}.}
[.3\linewidth]{\begin{tikzpicture}[use Hobby shortcut,
every trefoil component/.style={thick, draw}, pics/arrow/.style={code={%
  \draw[line width=0pt,{Computer Modern Rightarrow[line
  width=1pt,width=3ex,length=2ex]}-] (-0.5ex,0) -- (0.5ex,0);
  }}]
\path[spath/save=trefoil] ([closed]90:2) foreach \k in {1,...,3} {
.. (-30+\k*240:.5) .. (90+\k*240:2) } (90:2);
\tikzset{spath/knot={trefoil}{0pt}{1,3,5}}
\node[text=black] at (0,2) {$\boldsymbol{>}$};

\node[text=black] at (0.9,0.5) {\tiny{$\bullet$}};
\node[text=black] at (0.6,0.6) {\tiny{$\bullet$}};

\node[text=black] at (0.2,-0.8) {\tiny{$\bullet$}};
\node[text=black] at (0,-1) {\tiny{$\bullet$}};

\node[text=black] at (-0.9,0.5) {\tiny{$\bullet$}};
\node[text=black] at (-0.8,0.2) {\tiny{$\bullet$}};
\end{tikzpicture}
}
 \hspace{.03\textwidth}
\subcaptionbox{Black directed graph.}
[.3\linewidth]{
\begin{tikzpicture}
\tikzset{nnode/.style = {shape=circle,fill=myg,draw,inner sep=1.5pt, minimum
size=0.2em}}
\tikzset{edge/.style = {->,> = stealth}}

\node[nnode] (S1) {};
\node[shape=coordinate,right=0.5\mynodespace of S1] (K) {};
\node[nnode,right=\mynodespace of S1] (S2) {};
\node[nnode,above=0.8\mynodespace of K] (S3) {};

\draw[edge,bend left=0] (S3)  to node{} (S1);

\draw[edge,bend left=0] (S1)  to node{} (S2);

\draw[edge,bend right=0] (S2)  to node{} (S3);

\end{tikzpicture}}
 \hspace{.03\textwidth}
\subcaptionbox{White directed graph.}
[.3\linewidth]{\begin{tikzpicture}
\tikzset{nnode/.style = {shape=circle,fill=myg,draw,inner sep=1.5pt, minimum
size=0.2em}}
\tikzset{edge/.style = {->,> = stealth}}

\node[nnode] (S1) {};
\node[nnode,right=\mynodespace of S1] (S2) {};

\draw[edge,bend left=0] (S1)  to node{} (S2);

\draw[edge,bend left=40] (S1)  to node{} (S2);

\draw[edge,bend right=40] (S1)  to node{} (S2);

\end{tikzpicture}}
\caption{The black and white directed graphs of the Tait diagram of the oriented trefoil knot depicted in Figure \ref{fig:trefoil}. We consider the checkerboard coloring where the outside region is colored with white. \label{fig:signedgraph}}
\end{figure}
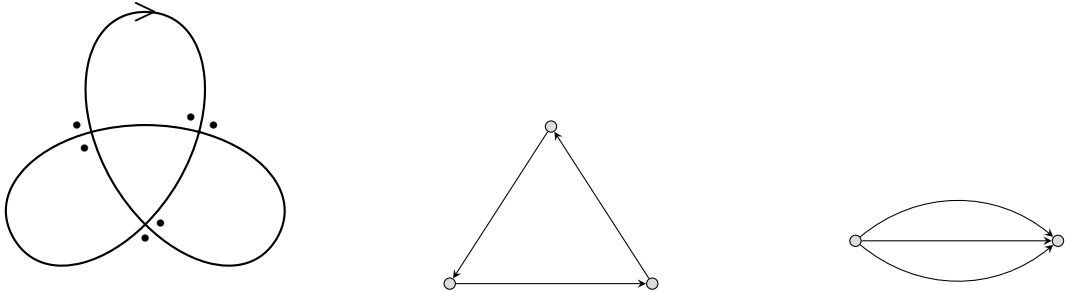

We define \textit{graph codes} from such directed graphs using their incidence matrices.

\begin{definition}
\label{def:graphcode}
Let $\Gamma$ be a directed graph, $v_1,\ldots,v_m$ be an enumeration of the vertices of the graph, and $e_1,\ldots ,e_n$ be an enumeration of the edges of the graph. Let $R$ be a ring and let $t\in R$ be an invertible element of $R$. Define $A(t)$ be the $m \times n$ matrix with entries:
$$
A(t)_{ij} = 
\begin{cases}
1 & \text{if} \ \ e_j \textnormal{ is an outgoing edge of } v_i,\\
t & \text{if} \ \ e_j \textnormal{ is an ingoing edge of } v_i,\\
0 & \text{otherwise}.
\end{cases}
$$
Then, the $R$-linear code with parity check matrix $A(t)$ is denoted by $\mC_{\Gamma,t}$.
\end{definition}

\begin{remark}\label{r-A(t)}
If $t=-1$ in Definition \ref{def:graphcode}, then $A(-1)$ is the incidence matrix of $\Gamma$ and has the property that the sum of the rows is the all-zero vector.
So, deleting a row of $A(-1)$ gives a matrix that is still a parity check matrix of $\mC_{\Gamma,-1}$.
The code $\mC_{\Gamma,-1}$ is abbreviated by $\mC_{\Gamma}$ and is called the \textbf{cycle code} of $\Gamma $,
and its dual is called its \textbf{graph code}.
Sometimes the cycle code is called graphic or cographic, see \cite[\S 8.1.2]{pellikaan2018}. 
\end{remark}

We note that the graph code is denoted by $C_{\Gamma}$ in \cite[\S 8.1.2]{pellikaan2018}. It corresponds to the notation $\mC_{\Gamma}^\perp $ in this paper.

\begin{remark}\label{r-param-cyclecode}
Let $\Gamma$ be a connected graph with $V$ vertices and $n$ edges.
The cycle code $\mC_{\Gamma}$ of $\Gamma $ is an $[n,k,d]$ code,
where $k=n-V+1$ and $d$ is the girth, the length of the smallest cycle, of $\Gamma$, see \cite[Proposition 8.1.22]{pellikaan2018}. 
\end{remark}

Throughout this section, we point out whether the defined codes are asymptotically good or not, and propose an open question at the end, see Definition \ref{def:asympgood}. Graphic and cographic codes are not asymptotically good \cite{kashyap2008}.

Definition \ref{def:signedgraph} motivates the following definition.

\begin{definition}
\label{def:blackcode}
Let $R$ be a ring and let $t\in R$ be an invertible element of $R$.
Let $\Gamma _D$ and $\Gamma_{D^*}$ be the black and white directed graphs of the Tait diagram $D$ of a knot. 
The codes $\mC_{\Gamma_D,t}$ and $\mC_{\Gamma_{D^*},t}$ of these graphs are called \textbf{black code} and \textbf{white code}, 
and denoted by $\mC_{D,t}$ and $\mC_{D^*,t}$, respectively.
And $\mC_{D,-1}$ and $\mC_{D^*,-1}$ are abbreviated by $\mC_{D}$ and $\mC_{D^*}$, respectively.
\end{definition}

Let $B_D$ be the incidence matrix of the black directed graph $\Gamma_D$ of $D$ and let  $W_D$ be the incidence matrix of the white directed graph $\Gamma_{D^*}$ of $D$.
Let $\mathbf{b}$ be a row of $B_D$ corresponding to the black region $B$. Then, the entries of $\mathbf{b}$ correspond to edges of $\Gamma _D$ which correspond to crossings of $D$. The entry is $0$ if the crossing is not in the boundary of $B$, it is $1$ if the crossing is in the boundary of $B$ and there is no dot in $B$ near that crossing, and it is $-1$ if the crossing is in the boundary of $B$ and there is no dot in $B$ near that crossing. Similarly, we do the same for $W_D$. These matrices are the parity check matrices of the black and white codes, respectively.

 \begin{theorem}
 \label{p-dual-black-white}
 Let $D$ be a reduced Tait diagram of a knot.
 If the characteristic is $2$ or the diagram is alternating, 
 then the black and white codes are dual to each other, i.e., $\mC_D^\perp =\mC_{D^*}$.
 \end{theorem}
 
 \begin{proof}
 Let $v$ be a crossing in the the intersection of the boundaries of a black and a white region of $D$.
 Then $v$ is a crossing of $D$ and it lies on a piece of a strand, call it $e$, between $v$ and another crossing $v'$ and 
 that is in the boundary of both a black and a white region.
 Then $v\not=v'$, otherwise $e$ can be deformed such that one get a loop that it is not self-intersecting and is in the interior of one the regions except $v$. So we get an unknot that intersects the diagram $D$ in exactly $v$, that means that $v$ is a reducible crossing which contradicts the assumption that $D$ is reduced. Hence $e$ is not a loop and there is a unique crossing $v'$ which is distinct from $v$ and is incident to $e$. In particular, in the the intersection of the boundaries of a black and a white region of $D$ the number of crossings is even.
 
 Let $\mathbf{b}_i$ be a row $B_D$ corresponding to the black region $B_i$ and $\mathbf{w}_j$ a row of $W_D$ corresponding to the white region $W_j$. If the characteristic is $2$, then
 \begin{equation*}
 \mathbf{b}_i \cdot \mathbf{w}_j =  \sum_{v \in \partial B_i \cap \partial W_j}1 = 0
 \end{equation*}
 is equal to 0 since $|\partial B_i \cap \partial W_j |$ is even, proving the result.
 Now, suppose that the diagram is alternating. 
 If $\mathbf{b}_i \cdot \mathbf{w}_j $ has a nonzero contribution at a crossing $v$ in the summation, 
 then the crossing is in the intersection of the boundaries $B_i$ and $W_j$. 
 The crossings appear in pairs, so  there are distinct crossings $v$ and $v'$ that are endpoints of the piece of a strand $e$ that is contained  
 $ \partial B_i \cap \partial W_j$. 
 Since the diagram is alternating, we may assume that $e$ is part of an overcrossing at $v'$ and of an undercrossing at $v$.
 Suppose that $B_i$ is on the right-hand side of $e$ and $W_j$ is on the left-hand side of $e$. (Similar reasoning follows if it is the other way around.) Then the entry of $ \mathbf{b}_i$ at $v'$ is $1$ and the entry of $ \mathbf{w}_j$ is $-1$, since $e$ is part of an overcrossing at $v'$. 
 So the contribution to the inner product is $1\cdot (-1)=-1$. 
 The entries of $ \mathbf{b}_i$ and $\mathbf{w}_j$ at $v$ are both $1$ or both $-1$, since $e$ is part of an undercrossing at $v$.
 So the contribution to the inner product is $1$ in that case. Hence, the nonzero contributions to $ \mathbf{b}_i \cdot \mathbf{w}_j$ appear in pairs of $\pm 1$,
 and they sum up to zero. Therefore $\mC_D \perp \mC_{D^*}$.
 
 Suppose that the diagram $D$ consists of $n$ crossings and $b$ black regions, then $\Gamma_D$ is a graph with $n$ edges and $b$ vertices.
 Hence $\mC_D $ has length $n$ and dimension $n-b+1$ by Remark \ref{r-param-cyclecode}. 
 The total number of regions is $n+2$ by Lemma \ref{lem:knot_diag}. So the number of white regions is $n+2-b$.
 Hence $\Gamma_{D^*}$ is a graph with $n$ edges and $n+2-b$ vertices. 
 Therefore $\mC_{D^*}$ has length $n$ and dimension $n-(n+2-b)+1=b-1$ by Remark \ref{r-param-cyclecode}. 
 Hence the codes $\mC_D$ and $\mC_{D^*}$ have complementary dimensions. Therefore $\mC_D^\perp =\mC_{D^*}$, concluding the proof.
 \end{proof}

Theorem \ref{p-dual-black-white} does not generalize to the case of arbitrary $t$, since in general $\mC_{D,t}$ and $\mC_{D^*,t}$ do not have complementary dimensions, and they are not perpendicular to each other. The fact that the proof of Theorem \ref{p-dual-black-white}  works for $t=-1$ boils down to two facts:
\begin{enumerate}
    \item The sum of rows of the parity check matrix of the black graph is the all-zero vector. The same holds for the white graph. So the corresponding codes have complementary dimensions,
    \item The inner product of a row of the parity check matrix of the black graph with a row of the parity check matrix of the white graph is zero.
\end{enumerate}

\begin{proposition}\label{p-AB-hull}
 The code $\mC_{D,t} \cap \mC_{D^*,t}$ is equal to the Alexander-Briggs code $\mA_{D,t}$. 
 If $t=-1$, then $\mA_{D}$ is equal to the hull of $\mC_{D}$.
 \end{proposition}

\begin{proof}
The Alexander-Briggs code $\mA_{D,t}$ is defined by the parity checks defined by both the black and white regions.
Hence $\mA_{D,t}=\mC_{D,t} \cap \mC_{D^*,t}$. If $t=-1$, then $\mC_{D^*}= \mC_{D^*,-1}=\mC_{D}^\perp $ by Theorem \ref{p-dual-black-white}.
Hence $\mA_{D}$ is the hull of $\mC_{D}$.
\end{proof}

Combining Theorem \ref{p-dual-black-white} and Proposition \ref{p-AB-hull}, we get the next result related to LCD codes.

\begin{corollary}
\label{c-dual}
Let $D$ be a reduced Tait diagram of a knot.
If the characteristic is $2$ or the knot is alternating, then the Alexander-Briggs code $\mA_{D}$ (when $t=-1$) is LCD.
\end{corollary} 

We add the next remark about LCD codes and whether graph codes of Tait diagrams of knots can lead to ``good" LCD codes.

\begin{remark}
If $\mC$ and $\mD$ are $(\pm1)$-permutation equivalent codes, then their hulls (see Definition \ref{def:lcd}) are also $(\pm1)$-permutation equivalent. This is not true for monomial equivalent codes. If $q>3$, then every linear code is monomial equivalent to an LCD code \cite{carlet2018}. 
So the question about the existence of LCD codes is the same as the question about the existence of linear codes in the case of $q>3$. 
However, the cases $q=2$ and $q=3$ need separate attention, see \cite{dougherty2017}. 
It was shown that that LCD codes are asymptotically good \cite{massey1992}, in fact they attain the Gilbert-Varshamov bound \cite{sendrier2004}. However, the graph codes of Tait diagrams of knots cannot give ``good" LCD codes since cycle codes are not asymptotically good as mentioned before. 
\end{remark}

We conclude the section with an open problem. 

\begin{oproblem}
Do Alexander-Briggs codes of knots give asymptotically good
codes?
\end{oproblem}

\section{Connected Sum of Knot Diagrams}
\label{sec:connected}
\label{sec:5}

Using the \textit{connected sum} operation, two knot diagrams form a new knot diagram. This will give us a way 
of constructing Fox knot codes
with arbitrary dimension. 
This section is devoted to studying how the codes of two knot diagrams are related to the code of their connected sum.

\begin{definition}
\label{def:knotsum}
The \textbf{connected sum}
of oriented knots $K_1$ and $K_2$
is the oriented knot $K_1 \# K_2$ 
whose diagram is obtained by taking an arc from a strand of each knot and connecting the open ends with two new arcs, in such a way that the orientation is preserved in the sum; see Figure \ref{fig:knotsum}. In this way we get a diagram $D_1 \# D_2$ of $K_1 \# K_2$, where $D_1$ and $D_2$ are the diagrams of $K_1$ and $K_2$, respectively. \end{definition}

It can be shown that the 
connected sum of knots indeed
does not depend on the choice of the strands.

\begin{figure}[H]
    \centering
    \includegraphics[width=0.6\textwidth]{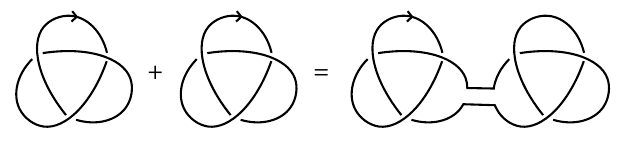}
    \includegraphics[width=0.6\textwidth]{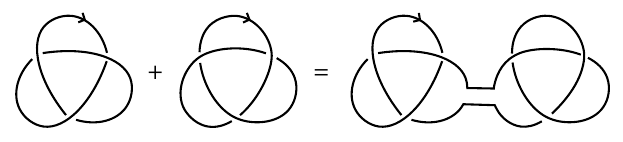}
    \caption{On the top, the composition of two trefoil knots results in what is called a so-called \textit{granny knot}. On the bottom, the composition of trefoil knot and its mirror image results in what is called a \textit{square knot}.}
    \label{fig:knotsum}
\end{figure}

The following concepts naturally arise from Definition \ref{def:knotsum}.

\begin{definition}
\label{def:prime}
A knot that cannot be written as the sum of two non-trivial knots is called a \textbf{prime} knot, otherwise it is called a \textbf{composite} knot.
\end{definition}

All composite knots have a unique decomposition into prime knots \cite{schubert2013}. Determining whether a knot is composite or not is generally a hard task.

We establish the notation for the rest of this section. 

\begin{notation}
\label{not:sum}
In the sequel we let
$D_1$ and $D_2$ be knot diagrams of (oriented, polygonal) knots $K_1$ and~$K_2$ with strands $x_1, ..., x_n$ and $y_1, ..., y_m$, respectively. 
We let $\mF_1$ and $\mF_2$ be their respective codes in~$\mathbb{F}_{q}^n$ and~$\mathbb{F}_{q}^m$, as in Definition~\ref{def:knotcode},
where $p$ is prime and $a$ is a positive integer. Moreover, we let 
$$\mF_1' = \{ c \in \mF_1 \mid c_n = 0 \}, \ \ \mF_2' = \{ d \in \mF_2 \mid d_m = 0 \}.$$
\end{notation}

The following 
result provides an explicit description of
the connected sum of knot diagrams. 

\begin{lemma}\label{lem:colorknotsum}
The Fox code of the sum $K_1\#K_2$ taken by connecting the knots diagrams~$D_1$ and~$D_2$, respectively over strands $x_n$ and $y_m$ is 
$$\mF_1\# \mF_2 = \{(c,d) \mid c \in \mF_1, d \in \mF_2, c_n = d_m\}.$$
\end{lemma}

\begin{proof}
A Fox coloring for $D_{K_1\#K_2}$ consists of a Fox coloring of $D_1$ and a Fox coloring of $D_2$ where the colors of the strands $x_n$ and $y_m$ that have been connected have the same color.
\end{proof}

Lemma \ref{lem:colorknotsum} implies that if $\mF_1$ and $\mF_2$ have parity check matrices $H_{\mF_1}$ and $H_{\mF_2}$, then $\mF_1\# \mF_2$ has parity check matrix

\begin{equation}
\label{eq:paritycheck}
  H_{\mF_1 \# \mF_2} = \left(
\begin{array}{cccc|cccc}
& H_{\mF_1} & & & & 0_{(n - \dim(\mF_1)) \times m} &  & \\
& 0_{(m - \dim(\mF_2)) \times n} & & & & H_{\mF_2} & & \\
0 & \dots & 0 & 1 & 0 & \dots & 0 & -1
\end{array} \right).  
\end{equation}
The last line of the matrix consists only of zeroes, except for a 1 on the $n$-th column and -1 on the $(n+m)$-th column.

The following proposition shows that the connected sum gives us another method, besides pretzel knots (see Theorem \ref{thm:pretzeldim}), to construct codes with any dimension.

\begin{proposition}
\label{prop:dimsum}
We have
$$\dim(\mF_1 \# \mF_2) = \dim(\mF_1) + \dim(\mF_2) -1.$$
\end{proposition}
\begin{proof}

The parity check matrices $H_{\mF_1}$ and $H_{\mF_2}$ of $\mF_1$ and $\mF_2$ are of size $(n - \dim(\mF_1)) \times n$ and $(m - \dim(\mF_2)) \times m$, respectively. Using the above construction from matrix \eqref{eq:paritycheck} we then get a parity check matrix $H$ for $\mF_1\# \mF_2$ of size~$(n+m - (\dim(\mF_1) + \dim(\mF_2) - 1)) \times (n + m)$ such that the first $n + m - (\dim(\mF_1) + \dim(\mF_2))$ rows are linearly independent. 

Towards a contradiction, assume that the last row can be written as a linear combination of the other rows of $H$. That would mean that there exists a linear combination of the rows of $H_{\mF_1}$ equal to $(0,...,0,1)$, which means that the strand~$x_m$ should always be colored with 0. The possible trivial colorings contradict this, as these include vectors with the same nonzero element on each position. So we find that the last row of $H_{\mF_1\# \mF_2}$ is independent from the other rows. Therefore, the rank of the matrix is $n + m - (\dim(\mF_1) + \dim(\mF_2) - 1)$. This proves the 
desired result.
\end{proof}

The diagram of the $m$-fold sum construction of the trefoil knot gives a code over $\F_3$ of length $3m$ and dimension $m+1$. 
Hence its rate is $R=(m+1)/3m \approx 1/3$. 

As one expects, the Alexander polynomials of two knots and their knot sum are also related; see~\cite[Theorem~6.3.5]{murasugi1996knot}.

\begin{proposition}
\label{prop:detsum}
We have 
$$
\Delta_{K_1\#K_2}(t) = \Delta_{K_1}(t)\Delta_{K_2}(t)
.$$
\end{proposition}

Next, we give an example of a connected sum of two knot diagrams and compute the determinant using Proposition \ref{prop:detsum}.

\begin{example}
\label{ex:detsum}
In Figure \ref{fig:sum3141}, the diagrams of the trefoil knot, figure-eight knot and their connected sum are depicted. 

\begin{figure}[H]
    \centering
    \includegraphics[width=\textwidth]{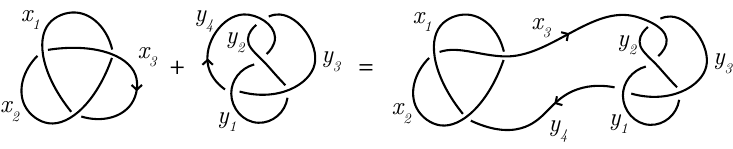}
    \caption{}
    \label{fig:sum3141}
\end{figure}

\noindent Using the matrix in \eqref{eq:paritycheck}, their coloring matrices are derived as follows:
$$
\begin{pmatrix}
 1 & 1 & -2 \\ 
 -2 & 1 & 1 \\
  1 & -2 & 1 \\
\end{pmatrix} \textnormal{, }
\begin{pmatrix}
  1 & 1 & -2 & 0 \\
  0 & 1 & 1 & -2 \\
  -2 & 0 & 1 & 1 \\
  1 & -2 & 0 & 1 \\
  \end{pmatrix},
  \textnormal{ and }
$$
\begin{equation}\label{eq: first matrix}
  \begin{pmatrix}
  1 & 1 & -2 & 0 & 0 & 0 & 0 \\
  -2 & 1 & 1 & 0 & 0 & 0 & 0 \\
  0 & 0 & 1 & 1 & -2 & 0 & 0 \\
  0 & 0 & 0 & 1 & 1 & -2 & 0 \\
  0 & 0 & -2 & 0 & 1 & 1 & 0 \\
  0 & 0 & 0 & -2 & 0 & 1 & 1 \\
  1 & -2 & 0 & 0 & 0 & 0 & 1
\end{pmatrix},  
\end{equation}
which are the parity check matrices of $\mF_1$, $\mF_2$ and $\mF_1\# \mF_2$, respectively. These knots have determinants 3, 5 and 15 by Proposition \ref{prop:detsum}, respectively. 
\end{example}

Lastly, we prove that the minimum distance of a code of the diagram of a connected sum is determined by the weight distributions of the codes of the constituent knot diagrams.

\begin{theorem}
\label{prop:sumknotmindist}
Let $\mF_1$, $\mF_1'$, $\mF_2$, and $\mF_2'$
and $\mF_1 \# \mF_2$ be as in Notation \ref{not:sum} and Lemma \ref{lem:colorknotsum} and 
let $d(\mF_1)$, $d(\mF_2)$ and $d(\mF_1 \# \mF_2)$ their respective minimum distances.  
The minimum distance of $\mF_1 \# \mF_2$ is equal to 
\begin{equation}
    \min \{ \ d(\mF_1'),\ d(\mF_2'),\ v+w \ \mid \  v \in  \wt(\mF_1 \setminus \mF_1'),\ w \in \wt(\mF_2 \setminus \mF_2') \ \}.
\end{equation}
\end{theorem}

\begin{proof}
As the codewords of $\mF_1$ and the codewords of $\mF_2$ only affect the weight of the codewords of $\mF_1 \# \mF_2$ at the first $n$ positions and the last $m$ positions, respectively, we look how minimum weight codewords of $\mF_1 \# \mF_2$ can be constructed by concatenating codewords of $\mF_1$ and $\mF_2$. 
Now $x \in \mF_1 \# \mF_2$ if and only if $x=(c,d) $ with $c \in \mF_1$ and $d \in \mF_2$  and $c_n = d_m$. Then $\wt(x)=\wt(c)+ \wt(d)$. We investigate two cases to finish the proof.
\begin{enumerate}
    \item Let $c_n = d_m = 0$, that is $c \in \mF_1'$ and $d \in \mF_2'$. 
In this case $\min\{ d(\mF_1'), d(\mF_2') \}$ is the smallest nonzero weight and is obtained by means of $(c, 0^m)$ or $(0^n,d)$ with the all-zeros codeword $0^n \in \mF_1$ and the all-zeros codeword $0^m \in \mF_2$.
    \item Let $c_n = d_m \neq 0$, that is $c \in \mF_1 \setminus \mF_1'$ and $d \in \mF_2 \setminus  \mF_2'$.
If $v =\wt (c)$, then $v \in \wt(\mF_1 \setminus \mF_1')$. Similarly, if $w =\wt (d)$, then $w \in \wt(\mF_2 \setminus \mF_2')$.
Conversely, if $v \in \wt(\mF_1 \setminus \mF_1') $, then there exists a  $c \in \mF_1 \setminus \mF_1'$ with $v =\wt (c)$. 
Similarly, if $w \in \wt(\mF_2 \setminus \mF_2')$, then there exists a  $d \in \mF_2 \setminus \mF_2'$ with $w =\wt (d)$. 
Hence, $\min \{ v+w \mid v \in  \wt(\mF_1 \setminus \mF_1'),\ w \in \wt(\mF_2 \setminus \mF_2')\}$ 
is the smallest weight of a nonzero codeword of $\mF_1 \# \mF_2$ obtained in this case. \qedhere
\end{enumerate}
\end{proof}

We give the following two remarks related to Theorem \ref{prop:sumknotmindist}.

\begin{remark}
If $\mF_1' =\{0\}$ and $\mF_2' =\{0\}$, then $K$ and $L$ have only trivial colorings, so $ \mF_1 \# \mF_2$ has only trivial colorings and $d(\mF_1 \# \mF_2)= n+ m$.
This is in agreement with the statement in Theorem \ref{prop:sumknotmindist}, since we defined the minimum distance of the zero code to be $\infty$ in Definition~\ref{def:weight_distance}.

If $\mF_1' =\{0\}$ and $\mF_2' \not=\{0\}$, then $d(\mF_1 \# \mF_2)=\min \{ \   d(\mF_2'),\ n+w \ \mid \ w \in \wt(\mF_2 \setminus \mF_2')\  \}$.\\
And a similar formula holds in case $\mF_1' \not=\{0\}$ and $\mF_2' =\{0\}$.
\end{remark}

\begin{remark}
Theorem \ref{prop:sumknotmindist} also follows from~\cite[Proposition 6.1.1]{nijsten2022knots}, 
where a formula for the weight enumerator of $\mF_1 \# \mF_2$ is given in terms of the weight enumerators of $\mF_1 $, $\mF_2$, $\mF_1'$, and~$\mF_2'$.
The formula is: 
$$
W_{\mF_1 \# \mF_2}(t) = W_{\mF_1'}(t)\cdot W_{\mF_2'}(t) + \frac{1}{q - 1} (W_{\mF_1}(t) - W_{\mF_1'}(t))(W_{\mF_2}(t) - W_{\mF_2'}(t)).
$$
This is in agreement with Theorem \ref{prop:sumknotmindist}, since $a_v(\mF_1') < a_v(\mF_1)$ if and only if $ v\in \wt(\mF_1 \setminus \mF_1')$,
and $a_w(\mF_2') < a_w(\mF_2)$ if and only if $w \in \wt(\mF_2 \setminus \mF_2')$.
\end{remark}

The next example shows applications of some of our results in this section.
\begin{example}
Let $q=3$ and $\mF=\mF_D$ where $D$ is the knot diagram of the trefoil knot depicted in Figure \ref{fig:trefoil}. By Example \ref{ex:detsum}, its parity check matrix is all-ones matrix and thus has rank 1. Thus, $\mF$ is a $[3,2,2]_3$ MDS code such that
\begin{align*}
    \mF &= \{(0,0,0),(0,1,2),(0,2,1),(1,0,2),(1,2,0),(1,1,1),(2,0,1),(2,1,0),(2,2,2)\}, \\
    \mF' &= \{(0,0,0),(1,2,0),(2,1,0)\}.
\end{align*}
We have $d(\mF')=2$ and $\wt(\mF \setminus \mF') = \{2,3\}$. By Proposition \ref{prop:dimsum} and Theorem \ref{prop:sumknotmindist} one can create a code with parameters
$$[n+m, \dim(\mF_2)+1, \min\{2,d(\mF_2')\}]_3$$
where $\mF_2$ is a code a knot diagram of some knot with $m$ strands. For example, if $\mF_2 = \mF_1$, then one gets a $[6,3,2]_3$ code which is a non-MDS. \end{example}

For the remaining part of this section, we focus on cycle codes.

\begin{definition}\label{d-sum-graph}
Let $\Gamma$ and $\Sigma$ be two (directed) graphs.
The \textbf{disjoint sum} of $\Gamma$ and $\Sigma$ is denoted by $\Gamma \sqcup\Sigma$
and has as nodes the disjoint union of the nodes of $\Gamma$ and $\Sigma$, 
and has as edges the disjoint union of the edges of $\Gamma$ and $\Sigma$.

Let $p$ be a node of $\Gamma$, 
and let $q$ be a node of $\Sigma$.
Then $(\Gamma \sqcup\Sigma )/ (p,q)$ is the graph
$\Gamma \sqcup\Sigma$ where the the node $p$ is identified with $q$.
\end{definition}

\begin{proposition}\label{p-sum-code-graph-pol}
Let $\Gamma_1$ and $\Gamma_2$ be two directed graphs. Let $p_1$ be a node of $\Gamma_1 $ and let $p_2$ be a node of $\Gamma_2 $.
Let $\Gamma =(\Gamma_1 \sqcup \Gamma_2 )/ (p_1,p_2)$.
Then 
$$
C_{\Gamma}=C_{\Gamma_1} \oplus C_{\Gamma_2}
$$  
\end{proposition}
\begin{proof}
Let $A_1$, $A_2$ and $A$ be the matrices of the directed graphs $\Gamma_1$,  $\Gamma_2$ and $\Gamma$, respectively as given in 
Definition \ref{def:graphcode} for $t=-1$. Then $A_1$, $A_2$ and $A$ are parity check matrices of the cycle codes 
 $C_{\Gamma_1}$, $ C_{\Gamma_2}$ and $C_{\Gamma}$, respectively by definition.
Let $A_1'$ be the matrix obtained from $A_1$ by deleting the row corresponding to $p_1$.
Let $A_2'$ be the matrix obtained from $A_2$ by deleting the row corresponding to $p_2$.
Let $A'$ be the matrix obtained from $A$ by deleting the row corresponding to $p_1=p_2$.
Then $A_1'$, $A_2'$ and $A'$ are also parity check matrices of the cycle codes 
$C_{\Gamma_1}$, $ C_{\Gamma_2}$ and $C_{\Gamma}$, respectively by Remark \ref{r-A(t)}, since $t=-1$. 
Now 
 $$
 A'=
 \left(
\begin{array}{c c}
 A_1' & 0 \\
 0    & A_2' \\
 \end{array}
 \right)
 $$
This proves the proposition.
\end{proof}

The graph $\Gamma =(\Gamma_1 \sqcup \Gamma_2 )/ (p_1,p_2)$ varies with the choices of the nodes $p_1$ and $p_2$, that is they are in general not isomorphic. But its graph code is independent of the choices of $p_1$ and $p_2$.

\begin{proposition}\label{p-prod-sum-knots}
Let the black regions of the constituent knots $K_1$ and $K_2$, and their Tait diagrams $D_1$ and $D_2$, respectively 
be such that their unbounded regions are white. 
Let $D_1\# D_2$ be the Tait diagram of $K_1\# K_2$ where the regions $B_1$ and $B_2$ of $D_1$ and $D_2$, respectively, are glued together.
Let $\Gamma_1$ and $\Gamma_2$ be the graphs of the black regions of $D_1$ and $D_2$, respectively.
Then $ (\Gamma_1 \sqcup \Gamma_2 )/ (B_1,B_2)$ is the graph of the black regions of the diagram of the connected sum $K_1\# K_2$.
\end{proposition}
\begin{proof}
This is a direct consequence of the definitions.  
\end{proof}

As a result of Propositions \ref{p-sum-code-graph-pol} and \ref{p-prod-sum-knots} we see that the cycle code of the connected sum of two knots 
does not depend on the choice of the strands and regions where the constituent knots are glued together.

\section{Dual of Fox Knot Codes}
\label{sec:dual}
\label{sec:6}

It is a standard problem in coding theory to understand how properties of a code determine or influence properties of the dual code.
In this short section, we ask ourselves 
if the dual of a Fox code of a knot diagram is also a Fox code of a knot diagram. 
We start by proving a necessary but not sufficient condition for a dual code to be a knot diagram.
\begin{proposition}
\label{prop:pdividesn}
Let $\mF$ be a Fox code of a knot diagram with $n$ strands over $\F_q.$ Then~$q$ divides~$n$ if $\mF^\perp$ is a Fox code of a knot diagram. 
\end{proposition}
\begin{proof}
By Remark \ref{rem:dim1}, the Fox code $\mF$ of a knot diagram with $n$ strands has the $n$-times repetition code $$\langle \underbrace{(1,1,...,1)}_n \rangle$$ as a subcode. If  $\mF^\perp$ is a code of some knot diagram, then it should also have the $n$-times repetition code as a subcode. 
We have that $$\underbrace{(a,a,...,a)}_{n}\underbrace{(a,a,...,a)}_{n} {^\top} = na^2$$ for all $a \in \F_q^n$. So in order for both a code and its dual to have the $n$ times repetition code as their subcode, it must be 
that $na^2 = 0$ for all $a \in \F_q^n$. Therefore, we must have that $n$ is divisible by $q$, as desired.
\end{proof}

Using results we obtained on the dimension of a Fox code of a knot diagram in Subsection~\ref{subsec:dim}, we can obtain information about the dual code as well.

\begin{proposition}
\label{prop:duallimit}
Let $\mF$ be the Fox code of a knot diagram. If $\dim(\mF) < \frac{n-1}{2}$, then~$\mF^\perp$ is not monomial equivalent to the Fox code of a knot diagram.
\end{proposition}
\begin{proof}
If $\dim(\mF) < \frac{n-1}{2}$ then $\dim(\mF^\perp) > \frac{n+1}{2}$.
The result then follows from Theorem~\ref{thm:dimension-ineq}.
\end{proof}

This result can be used on composite knot diagrams to determine whether the duals of their codes are codes of knot diagrams.
\begin{proposition}
\label{prop:4componentdual}
Let $\mF = \mF_1 \# \mF_2 \# \cdots \# \mF_i$ be the Fox code of a diagram of $i$ composed knots~$K = K_1 \# K_2 \# \cdots \# K_i$, where each $K_j$ has $n_j$ strands in their corresponding knot diagrams. If $i \geq 4$, then $\mF^\perp$ is not a Fox code of a knot diagram. 
\end{proposition}
\begin{proof}
Let $n = n_1 + n_2 + \cdots + n_i$.
Using Proposition \ref{prop:dimsum} we find that 
\begin{align*}
    \dim(\mF_1 \# \mF_2 \#\cdots \# \mF_i) &= \dim(\mF_1) + \dim(\mF_2) + \cdots + \dim(\mF_i) - i + 1 \\
    &\leq \frac{n_1 + 1}{2} + \frac{n_2 + 1}{2} + \cdots + \frac{n_i + 1}{2} - i + 1 \\
    &= \frac{n - i}{2} + 1.
\end{align*}
Therefore $\dim(\mF) < \frac{n-1}{2}$ if $i \geq 4$ and the result follows from Proposition \ref{prop:duallimit}.
\end{proof}

\subsection*{Data Availability}
There is no data associated with this article. This article is self-contained. 
\subsection*{Competing Interest}
The authors have no conflicts of interest that could potentially influence or bias this article. 

\bigskip

\bibliographystyle{abbrv}
\bibliography{ADV}

\bigskip

\appendix

\section{Commutative Algebra} \label{commut-alg} 

For the basic definitions and properties of commutative algebra such as modules and morphisms we refer to \cite{atiyah2018introduction,eisenbud1995,lang2012algebra}.
In this paper, a ring will always mean a Noetherian, commutative ring with a unit element $1$. 
So, the ideals of a ring are finitely generated. Furthermore, all modules will be assumed to be finitely generated.

\begin{remark}\label{r-row-column}
In this appendix, we adopt the usual convention in commutative algebra to consider the elements of $R^{(n)}$ as column vectors of length $n$ with entries in $R$,
contrary to the rest of this paper where we align to the convention in coding theory where the elements of $R^n$ are row vectors of length $n$ with entries in $R$.
So this difference is stressed by the notation $R^{(n)}$ for column vectors and $R^n$ for row vectors.

The set of $m \times n$ matrices with entries in the ring $R$ is denoted by $R^{m \times n}$.
The matrix~$A \in R^{m \times n}$ gives a morphism of $R$-modules $R^{(n)} \rightarrow R^{(m)}$ 
defined by $x \mapsto Ax$ for $ x \in R^{(n)}$.
The \textbf{kernel} of $A \in R^{m \times n}$ is $\Ker (A) =\{ x\in R^n \mid Ax^T =0 \}$.
\end{remark}

To define equivalence of matrices, row/column operations are used.

\begin{definition}\label{d-matrix-equiv}
The \textbf{elementary row operations} on a matrix with entries in a ring are:
\begin{enumerate}
    \item interchanging rows,
\item adding a row to another row,
\item multiplying a row with an invertible element of the ring.
\end{enumerate}
\end{definition}

One defines \textbf{elementary column operations} similarly.
If $A$  is the $m \times n$ matrix in the left upper submatrix of the $(m+1)\times (n+1)$ matrix $B$ such that the entries of the last row and column of $B$ are all zero, except a pivot $1$ at the entry corresponding to the last row and last column, then we say that $B$ is obtained from $A$ by \textbf{adding a pivot}, and $A$ from $B$ by \textbf{deleting a pivot}.

\begin{definition}\label{d-equiv-matrix}
Matrices are called \textbf{equivalent} if they can be obtained from each other by a sequence of
\begin{itemize}
    \item elementary row and column operations,
    \item adding and deleting a zero row, 
    \item adding and deleting a pivot.
\end{itemize}
\end{definition}

Definition \ref{d-equiv-matrix} is taken from \cite[Chapter VII \S 4]{crowell2012introduction} and is more general than the one given in  \cite[Chapter II]{newman1972}, where  equivalent matrices must have the same size.

\begin{proposition}\label{p-module-equiv}
Let $A$ and $B$ be  matrices with entries in $R$.
If $A$ and $B$ are equivalent, then $\Ker (A) \cong \Ker (B)$ as $R$-modules.
\end{proposition}

Given a matrix, one defines ideals generated by the determinant of all submatrices of some fixed size.

\begin{definition}\label{d-elem-id}
Let $A \in R^{m \times n}$ and $k \in \Z_{\geq 0}$. 
Let $E_k(A)$ denote the $k$-th \textbf{elementary} (or \textbf{Fitting}) \textbf{ideal} of $A$, that is the ideal generated by determinants of all $(n-k)\times (n-k)$ submatrices of $A$ if $0<n-k\leq m$, $E_k(A)=0$ if $n-k>m$, and $E_k(A)=R$ if $n-k\leq 0$. 
\end{definition}

Elementary ideals of equivalent matrices are the same. Moreover this fact can be slightly refined,
as the following two propositions formalize.

\begin{proposition}\label{p-elem-id-1}
Let $A \in R^{m \times n}$ and $k \in \Z_{\geq 0}$. The elementary ideals $E_k(A)$ form an increasing sequence of ideals with respect to inclusion. If $A$ and $B$ are equivalent matrices, then $E_k(A)=E_k(B)$.
\end{proposition}
\begin{proof}
See \cite[Chapter VII (4.1)]{crowell2012introduction}.
\end{proof}

\begin{proposition}\label{p-free-matrix}
Let $A \in R^{m \times n}$ and $B \in R^{m \times (n+l)}$ be matrices such that $B$ is equivalent to $\smash{(A \mid O_{m\times l})}$, where 
 $0_{m\times l} \in R^{m \times l}$ is the matrix with all zero entries.
Then $E_{k}(B)=E_{k-l}(A)$ for all $k$.
\end{proposition}
\begin{proof}
The result follows directly from the definitions if $B=(A \mid O_{m\times l})$, and from Proposition \ref{p-elem-id-1} otherwise.
\end{proof}

Adding zero rows to a matrix does not change its elementary ideals. Thus, we have the following result that is independent of the number of columns of the matrix.

\begin{proposition}\label{p-elem}
Let $R$ be a field and let $A \in R^{m \times n}$. 
If $A$ has rank $r$, then  $E_k(A)=R$ if~$k\geq n-r$, and $E_k(A)=0$ otherwise.
\end{proposition}

\begin{proof}
If $A$ has rank $r$, then one can transform $A$ by elementary row and column operations into a matrix $B$ that has the $r \times r$ identity matrix $I_r$ as a submatrix and entries equal to zero outside that identity matrix. 
Deleting the $r$ rows and columns corresponding to the pivots of the  matrix gives the $(m-r) \times (n-r)$ matrix with zeros as entries.
The elementary ideals remain the same under these transformations by Proposition \ref{p-elem-id-1}.
Hence $E_k(A)=R$ if $k\geq n-r$ and $E_k(A)=0$ otherwise. 
\end{proof}

\begin{proposition}\label{p-elem-id-2}
Let $\varphi : R \rightarrow S$ be a morphism of rings and let $A$ be a matrix with entries $a_{ij}$ in $R$.
Denote by $\varphi (A)$ the matrix with entries $\varphi (a_{ij})$ in $S$. 
If $\varphi $ is surjective, then~$E_k (\varphi (A)) = \varphi (E_k(A))$.
\end{proposition}

\begin{proof}
See \cite[Chapter VII (4.3)]{crowell2012introduction}
\end{proof}

For the rest of the appendix, we focus on the principal ideals of a principal ideal domain~$R$ and its relations with the elementary ideals of a matrix whose entries are coming from $R$.

\begin{proposition}[{\bf Smith Normal Form}]\label{p-struct-pid}
Let $R$ be a principal ideal domain and let~$A$ be a matrix with entries in~$R$.
Then there is an increasing sequence of principal ideals~$(d_1)\subseteq (d_2)\subseteq \cdots \subseteq (d_l)\not= R $ 
such that $A$ is equivalent to a diagonal square matrix with~$(d_1, d_2, \ldots ,d_l) $ on the diagonal.
\end{proposition}
\begin{proof}
See \cite[Theorem II.9]{newman1972}.
\end{proof}

The principal ideals  $(d_i)$ in the previous proposition are called  \textbf{invariant factors} of the matrix $A$.
A generator of $(d_i)$ is unique up to an unit and the invariant factors are unique.
Note that the
principal ideals might be zero. 
Let $r$ be the smallest non-negative integer such that $d_{r}=0$ and $d_{r+1} \not= 0$,
where
$d_0=0$ and $d_{l+1}=1$. 
Then the smallest non-negative integer $r$ such that $d_{r}=0$ and $d_{r+1} \not= 0$
is called the \textbf{rank} of the matrix $A$.

\begin{corollary}\label{c-struct-pid}
Let $R$ be a principal ideal domain. Let $M$ be a matrix with entries in $R$
and invariant factors $(d_1)\subseteq (d_2)\subseteq \cdots \subseteq (d_l)$.
Then $E_k(M)$ is generated by $$\Delta_k:=\prod_{j=k+1}^l d_j.$$
Conversely, let $E_k(M)=(\Delta_k)$. Then $\Delta_{k-1}$ is divisible by  $\Delta_k$ and $d_k =\Delta_{k-1}/\Delta_k$ is the $k$-th invariant factor of $M$.
\end{corollary}
\begin{proof}
See \cite[Chapter II \S 15 and \S 16]{newman1972}.
\end{proof}

We conclude this appendix with the following proposition which in this paper is used in the principal ideal domains $R=\Z$ and $R=\F_p[T]$, and in their localizations; see Propositions~\ref{p-number-color-Z} and~\ref{p-number-color-prime}.

\begin{proposition}\label{p-struct-ker-im}
Let $R$ be a principal ideal domain. 
Let $A$ be a matrix with entries in~$R$ and invariant factors $(d_1)\subseteq (d_2)\subseteq \cdots \subseteq (d_l) $. 
Let $d$ be a nonzero element of $R$ and let~$a_i=\gcd (d,d_i)$ and $a_ib_i=d$.
Let $\overline{R}=R/(d)$ and $\overline{x}=x+(d) \in \overline{R}$ for $x\in R$. 
Then
$$
\Ker (\overline{A})  \cong \overline{R}/(\overline{a_1}) \oplus \overline{R}/(\overline{a_2}) \oplus  \cdots \oplus  \overline{R}/(\overline{a_l}).
$$
\end{proposition}
\begin{proof}
The matrix $A$ is equivalent to the diagonal matrix $B$ that has $(d_1,d_2, \ldots ,d_l)$ on its diagonal by Proposition \ref{p-struct-pid}.
Hence $\Ker (A ) \cong \Ker (B)$  by Proposition~\ref{p-module-equiv}.
To prove the result it is enough to show it separately for each $d_i$ on the diagonal. Notice that $\overline{d_i}=\overline{a_i}$, since $a_i=\gcd (d,d_i)$.
Consider the sequence of $\overline{R}$-modules:
$$
0\rightarrow (\overline{b_i})\overline{R} \rightarrow \overline{R} \rightarrow \overline{R} \rightarrow (\overline{b_i})\overline{R} \rightarrow 0
$$
where the map $(\overline{b_i})\overline{R} \rightarrow \overline{R}$ is an inclusion, and $\overline{R} \rightarrow \overline{R}$ is given by multiplication by $\overline{a_i}$,
and the surjective map  $\overline{R} \rightarrow (\overline{b_i})\overline{R} $ is given by multiplication by $\overline{b_i}$. 
This sequence is a chain complex, that is, the composition of two consecutive maps is zero, since~$a_ib_i=d \equiv 0 \pmod d$.
But it is in fact an exact sequence: Consider the kernel of the multiplication by~$\overline{a_i}$ and suppose that $\overline{x}\overline{a_i}=0$. Then $xa_i \equiv 0 \pmod d$, and thus $xa_i = yd$ for some $y\in R$. So~$xa_i = ya_ib_i$, and consequently $x = yb_i$ since $R$ is an integral domain. Therefore $\overline{x} \in (\overline{b_i})\overline{R}$. 

On the right hand, we have the sequence $ \overline{R} \rightarrow \overline{R} \rightarrow (\overline{b_i})\overline{R}$, 
which is exact at the middle by a similar reasoning as before.
The cokernel of the multiplication by $\overline{a_i}$ is by definition equal to $\overline{R}/(\overline{a_i})$.
Hence $\overline{R}/(\overline{a_i})$ isomorphic to $(\overline{b_i})\overline{R}$. Therefore the statement on $\Ker (\overline{A})$ follows, as desired.
\end{proof}

\end{document}